\documentclass[12pt]{article}
\usepackage{amsmath}
\usepackage{graphicx}
\usepackage{enumerate}
\usepackage{natbib}
\usepackage{url} 

\allowdisplaybreaks
\newcommand{\blind}{0}

\usepackage{rotating}
\usepackage{color}

\usepackage{relsize}

\def\drt{{\scriptscriptstyle\mathrm{fps}}}
\def\dre{{\scriptscriptstyle\mathrm{eps}}}
 
\def\minone{{\scriptscriptstyle{X}}}  
\def\mintwoone{{\scriptscriptstyle{W,X}}}  
\def\mintwotwo{{\scriptscriptstyle{W,W}}}  

\def\mintwo{{\scriptscriptstyle{W}}}
\def\minzero{{\scriptscriptstyle\mathrm{0}}}

\usepackage{multirow}
\def\E{\mathbb{E}}

\def\dr{\scriptscriptstyle{dr}}

\def\T{{\mathrm{\scriptscriptstyle T} }}
\def\dr{{\mathrm{\scriptscriptstyle{}} }}

\def\ls{{\mathrm{\scriptscriptstyle{ls}} }}

\makeatletter
\newcommand*{\indep}{%
	\mathbin{%
		\mathpalette{\@indep}{}%
	}%
}

\newcommand*{\nindep}{%
	\mathbin{
		\mathpalette{\@indep}{\not}
	}%
}
\newcommand*{\@indep}[2]{%
	\sbox0{$#1\perp\m@th$}
	\sbox2{$#1=$}
	\sbox4{$#1\vcenter{}$}
	\rlap{\copy0}
	\dimen@=\dimexpr\ht2-\ht4-.2pt\relax
	\kern\dimen@
	{#2}%
	\kern\dimen@
	\copy0 
}
\makeatother
 \usepackage{booktabs}
 \usepackage{amssymb}
\newtheorem{example}{Example}

\newtheorem{theorem}{Theorem}
\newtheorem{lemma}{Lemma}
\newtheorem{proof}{Proof}

\newtheorem{assumption}{Assumption}

\newtheorem{fact}{Fact}

\usepackage{caption}
\usepackage{subcaption}
\usepackage{amsfonts}
\def\E{\mathbb{E}}

\usepackage{subcaption}  
\usepackage[all]{xy}
\usepackage{bm}
\usepackage{pifont}
\usepackage{stmaryrd}
\usepackage{pgf,tikz}
\usepackage{tikz}
\usetikzlibrary{shapes,arrows,positioning,calc}

\usepackage[flushleft]{threeparttable}
\usepackage{enumitem}
\usepackage{mathrsfs}

\def\pr{\textnormal{pr}}

\usepackage{arydshln}

\def\pr{\textnormal{pr}}

\def\T{{ \mathrm{\scriptscriptstyle T} }}

\newcommand\cmtt{\fontfamily{cmtt}\selectfont}
\usepackage{hyperref}

\usepackage[symbol]{footmisc}

\usepackage{xr}
\makeatletter
\newcommand*{\addFileDependency}[1]{
  \typeout{(#1)}
  \@addtofilelist{#1}
  \IfFileExists{#1}{}{\typeout{No file #1.}}
}
\makeatother

\begin{document}

	\def\spacingset#1{\renewcommand{\baselinestretch}%
		{#1}\small\normalsize} \spacingset{1}

	\newcounter{savecntr}
\newcounter{restorecntr}
		\if0\blind
	{ 


    	\title{\bf 
Efficiency-improved doubly robust estimation with non-confounding predictive covariates}
        
 			\author{ Shanshan Luo\textsuperscript{1}, Mengchen Shi\textsuperscript{1},  Wei Li\textsuperscript{2}\thanks{Corresponding author}, ~Xueli Wang\textsuperscript{1},  and      Zhi Geng\textsuperscript{1}     \\\\
			\textsuperscript{1} School of Mathematics and Statistics, \\Beijing Technology and Business University \\\\
  	\textsuperscript{2}  Center for Applied Statistics and School of Statistics, \\ Renmin University of China  }
    
		\date{}
		
		\maketitle
	} \fi
	
	\if1\blind
	{
		\bigskip
		\bigskip
		\bigskip
		\begin{center}
			{\LARGE\bf Title}
		\end{center}
		\medskip
	} \fi
		\medskip

\begin{abstract}

In observational studies, covariates with substantial missing data are often omitted, despite their strong predictive capabilities. These excluded covariates are generally believed not to simultaneously affect both treatment and outcome, indicating that they are not genuine confounders and do not impact the identification of the average treatment effect (ATE). 
In this paper, we introduce an alternative doubly robust (DR) estimator that fully leverages  non-confounding  predictive covariates to enhance efficiency,  while also allowing missing values in such covariates. 
 Beyond the double robustness property, our proposed estimator is designed to be more efficient than the standard DR estimator.  Specifically, when the propensity score model is correctly specified, it achieves the smallest asymptotic variance among the class of DR estimators, and brings additional efficiency gains by further integrating predictive covariates. 
Simulation studies demonstrate the notable performance of the proposed estimator over current popular methods. An illustrative example is provided to assess the effectiveness of right heart catheterization (RHC) for critically ill patients.

\end{abstract} 
 \bigskip
 \bigskip
 \bigskip
\noindent%
\vfill

\newpage
\spacingset{1.9} 
 \section{Introduction}

In observational studies, covariates that affect both  treatment   and  outcome   are often referred to as confounders \citep{Geng2002JRSSB}. In practice, it is commonly suggested to include all collected baseline covariates to serve as potential confounders, which essentially requires that the treatment variable satisfies the ignorability or unconfoundedness assumption \citep{rosenbaum1983assessing}. 
However, missingness in covariates is a pervasive issue in observational studies across biomedical and social sciences, posing a significant challenge for subsequent inference. 
For instance, in the widely used RHC dataset \citep{hirano2001estimation,cui2021semiparametric}, which focuses on evaluating the effectiveness of right heart catheterization in the intensive care unit for critically ill patients, some covariates exhibit missingness rates exceeding 50\%. While these missing covariates may affect either the treatment or outcome, prior research
has considered them not to simultaneously affect both variables and, therefore, they do not qualify as confounder variables. 
In such cases, existing literature often neglects covariates with a substantial amount of missing values and focuses solely on true confounders, proceeding with widely used estimators that also provide consistent estimates.

Various methods have been developed to estimate   average treatment effect (ATE) in the presence of potential confounders, including outcome regression, inverse probability weighting (IPW), and others \citep{rosenbaum1983assessing}. Among these, the augmented inverse-probability weighted (AIPW) estimator, proposed by \citet{robins1994estimation} and \citet{bang2005doubly}, and also referred to as the doubly robust (DR) estimator, is widely employed in observational studies. The DR estimator is known for its double robustness, ensuring consistency when either   outcome model or  propensity score  is correctly specified. 
In the presence of predictive covariates that are not confounders (i.e., variables that do not directly affect the treatment but may be correlated with the outcome), many prior studies have shown that employing such covariates can improve estimation efficiency \citep{hahn2004functional, lunceford2004stratification, brookhart2006variable, de2011covariate, franklin2015regularized, craycroft2020propensity, tang2023ultra}. For  example, \citet{lunceford2004stratification} theoretically reveal that introducing additional covariates into    IPW estimator can enhance estimation efficiency compared to  discarding them. \citet{craycroft2020propensity} demonstrate that including predictive covariates can result in a lower semiparametric efficiency bound compared to the one based on true confounders \citep{Hanh1998Econometrica,hirano2003efficient}. 
However, there are currently no answers  about whether introducing predictive covariates into  DR-type estimators can enhance estimation efficiency  in the presence of potentially misspecified working models.

In observational studies, the propensity score may be known by design and it is more likely to be correctly specified compared to the outcome model. We thus focus on a class of DR estimators that are designed to be more efficient than standard DR estimators when the propensity score is correct. These estimators achieve improved performance relative to existing DR estimators by directly minimizing the asymptotic variance.
In the context of outcome missing at random, \citet{cao2009improving} propose such a DR  estimator to estimate the mean of the outcome.  For optimal individual treatment regimes (ITRs), \citet{pan2021improved} introduce an enhanced DR estimator by directly maximizing the DR estimator of the expected outcome over a class of ITRs. 
However, both methods focus solely on the marginal mean outcome. It is currently uncertain how their ideas can be applied to estimate the mean difference between two  potential outcomes, i.e., the ATE. This is because minimizing the asymptotic  variance of   ATE estimators  may introduce additional interaction terms through nuisance models, posing challenges to achieving a similarly improved double robustness property.  Besides, it also remains to be investigated whether the proposed estimators can be further improved by incorporating predictive covariates. 

We address these questions in this paper and propose an improved  DR estimator for the  ATE that achieves the smallest variance among a class of DR estimators. 
We also explore the inclusion of additional predictive covariates, allowing for some missingness in such covariates. Our theoretical findings indicate that integrating these predictive covariates into the proposed estimator substantially improves estimation efficiency, as long as the propensity scores are correctly specified, even if the outcome models are misspecified.

The subsequent sections of this paper follow this structure. Section \ref{sec: notation-assumption} introduces the notation, framework, and assumptions. Sections \ref{sec:tps} and \ref{sec:com-analy} present the improved  DR  estimator when propensity scores are either fully fixed or estimated. We also conduct comparative analyses for different adjustment sets.  Section \ref{sim-studies} performs simulation studies designed to assess the finite sample performance of the proposed estimators. In Section \ref{sec:app}, we further illustrate the proposed  method using a real data. The paper concludes with a brief discussion in Section \ref{sec: dis}. Technical details are available in the supplementary material. 

 \section{Set up}
 \label{sec: notation-assumption} 
 \subsection{Notation and assumptions}
Assuming that there are $n$ individuals who are independently and identically sampled from a superpopulation of interest. We consider evaluating the causal effect of a treatment $Z$ on an outcome $Y$ subject to confounding by the observed covariates ${ X}$.    The potential outcomes under treatment and control are denoted as $Y_1$ and $Y_0$, respectively. The observed outcome $Y$ can be expressed as a combination of these potential outcomes: $Y = ZY_1 + (1-Z)Y_0$. We make the   stable unit treatment value  assumption (SUTVA), which states that there is only one version of potential outcomes for each individual, and there is no interference between units \citep{rubin1974estimating}. We are interested in estimating the average treatment effect (ATE), defined as $\tau = \E(Y_1 - Y_0)$. 

A fundamental problem of causal inference is that we can never observe both potential outcomes for a unit simultaneously.   To ensure the identifiability, the strong unconfoundedness or ignorability assumption is commonly employed \citep{rosenbaum1983assessing}.
\begin{assumption}
\label{assump: positivity}
(i) $Z\indep (Y_0,Y_1)\mid X$, (ii) $0< \pr(Z=1\mid X)<1$. 
\end{assumption} 
Assumption  \ref{assump: positivity}(i) states that the treatment assignment $Z$ is independent of the potential outcomes $(Y_0, Y_1)$ given the covariates ${ X}$, thereby eliminating any potential unobserved confounder between  treatment and outcome. In a randomized experiment, the ignorability  assumption  holds naturally, because $Z$ is independent of $(Y_0, Y_1, { X})$.  
Assumption \ref{assump: positivity}(ii) is crucial to ensure adequate overlap between the covariate distributions of the treatment and control groups \citep{rosenbaum1983assessing}.

 Define $e_{\minzero}(X)=\pr(Z=1\mid X)$ and $Q_{\minzero}({{ X}}, Z)=\E(Y\mid {{ X}},Z )$.    We posit a parametric model $e_{\minone}( {{ X}};{\alpha} )$ for the propensity score $e_{\minzero}(X)$, where ${\alpha} $ denotes the nuisance parameter. 
The subscript “$X$” is used to denote the propensity score model only using the covariates $X$, and other similar subscripts will be introduced in subsequent discussions. 
When propensity score $e_\minone(X;\alpha)$ is correctly specified, let ${\alpha}_{\minzero}$ denote the true parameter, and we have $e_\minzero(X)=e_\minone(X;\alpha_\minzero)$.  
Following \citet{robins1994estimation}, a DR   estimator can be constructed as follows,
\begin{equation}
    \label{eq: dr-X}  
         \begin{aligned} 
 {\tau}^{\dr}_{\minone}(\hat {\alpha},\hat{\beta}   )&=  
 \mathbb{P}_n  \begin{bmatrix}
 \dfrac{ZY-\{{Z-e_{\minone}(  {{ X}};\hat{\alpha})}\}Q^{\dr}_{\minone}( {{ X}},1;\hat{\beta})}{e_{\minone}( {{ X}};\hat{\alpha})} \\\addlinespace[1mm] - \dfrac{(1- Z)Y-\{{e_{\minone}( {{ X}};\hat{\alpha})}-Z\}Q^{\dr}_{\minone}( {{ X}},0;\hat{\beta})}{ 1 - e_{\minone}( {{ X}};\hat{\alpha})}  
 \end{bmatrix},
\end{aligned}   
\end{equation}
where $ \mathbb{P}_n (\cdot)$ denotes the empirical mean operator, $Q_{\minone}   ^\dr({{ X}}, Z;{\beta})$ represents a specified working model for $Q_{\minzero} ({{ X}}, Z )$,   $\hat\alpha$ and $\hat\beta$ are the estimated parameters.
 By adding an augmentation term that involves both estimated propensity scores and outcome models, the DR estimator   provides additional protection against model misspecification \citep{robins1994estimation,Scharfstein1999JASA, bang2005doubly}.  The proposed estimator  ${\tau}^{\dr}_{\minone}(\hat {\alpha},\hat{\beta}   )$ in \eqref{eq: dr-X}   is doubly robust in the sense that it consistently estimates ATE 
 as long as one of the nuisance working models is correctly specified, that is, $e_\minone(   {{ X}} ;{\hat\alpha}) \stackrel{p}{\rightarrow} e_{\minzero} ({ X}  )$ or $Q^{\dr}_{\minone}({ X},Z ; \hat{{{\beta}}}) \stackrel{p}{\rightarrow} Q_{\minzero}({ X},Z )$.

\subsection{Incorporating additional predictive covariates}

In practice, besides the true confounders $X$, there may exist some baseline covariates denoted as $W$ that have predictive power for the outcome. These variables are correlated with the outcome but are not directly associated with the treatment, except through the true confounders $X$ \citep{craycroft2020propensity}. The aforementioned conditional independence can be expressed as follows: 
 \begin{assumption} 
 \label{assump: negative-control}
$ Z\indep (  W,Y_0,Y_1)\mid   X$. \end{assumption}

 When combined with Assumption \ref{assump: positivity},  Assumption \ref{assump: negative-control} implies that genuine confounders $X$ sufficiently address the potential confounding relationship between $Z$ and $(W, Y_0, Y_1)$. This condition  also implies     $Z \indep (Y_0, Y_1) \mid (X, W)$. 
In practical scenarios, because $W$ only predicts the outcome variable $Y$ without having a direct effect on the treatment variable, it does not qualify as a genuine confounder \citep{Geng2002JRSSB}. Therefore, even in cases where covariates $W$ are missing, it does not affect the identifiability of the  ATE.
In such cases, various methods can be applied to impute missing values for predictive covariates $W$ \citep{rubin1996multiple}. This is especially relevant when enhancing the estimation efficiency of ATE in randomized experiments with missing covariates, as discussed by \cite{Zhao2022JASA}. However, for simplicity, this paper assumes that covariates $W$ are either fully observed or completely imputed.  Refer to Figure \ref{fig:prop-adj} for a graphical representation of these concepts. 
\begin{figure}[t]
    \centering 
        \begin{tikzpicture}[>=stealth, scale=1.75]
            \node[draw, circle] (X) at (0, 1) {$X$};
            \node[draw, circle] (Z) at (-1.5, 0) {$Z$};
            \node[draw, circle ] (Y) at (0, 0) {$Y$};
            \node[draw, circle ] (W) at (1.5, 0) {$W$}; 
            \draw[->] (X) -- (Z);
            \draw[->] (X) -- (Y);
            \draw[->] (Z) -- (Y);
            \draw[->] (W) -- (Y);
            \draw[->] (X) -- (W);
        \end{tikzpicture} 
    \caption{Observational studies with additional predictive covariates $W$.}
    \label{fig:prop-adj} 
\end{figure}
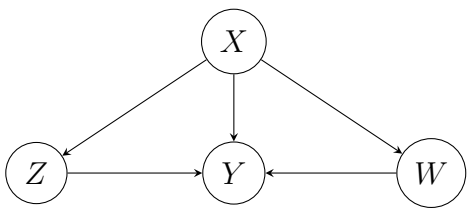

Therefore, we can also utilize the  set $(X, W)$ to construct the  DR-type estimators. In parallel with the previous discussion, we can specify a parametric model for $e_{\minzero}(X, W) = \pr(Z = 1 \mid X, W)$ as $e_{\mintwo}(X, W; \alpha, \gamma)$. Here, $\alpha$ continues to denote the parameters only related to the  terms with $X$, while $\gamma$ is an additional parameter vector corresponding to  $W$. The reason for this specification is that, given that $W$ does not directly affect $Z$, it is clear that $e_{\minzero}(X) = e_{\minzero}(X, W)$, and the true parameters of   $e_{\mintwo}(X, W; \alpha, \gamma)$ are $\alpha = \alpha_{\minzero}$ and $\gamma = \gamma_{\minzero} = 0$. Therefore, when we can correctly parametrize $e_{\minone}(X; \alpha)$, we can also correctly model $e_{\mintwo}(X, W; \alpha, \gamma)$.  For instance, we can propose the following pair of propensity score  models: $e_{\minone}(X ; \alpha ) =\mathrm{expit}(\alpha_1+\alpha_2X )$  and $
 e_{\mintwo}(X, W; \alpha, \gamma)=\mathrm{expit}(\alpha_1+\alpha_2 X+\gamma_1 W+\gamma_2 XW), $
where $\alpha=(\alpha_1, \alpha_2)$ and $\gamma=(\gamma_1, \gamma_2)$. 
When $\gamma$ takes its true value $0$, the model $e_{\mintwo}(X, W; \alpha, \gamma)$ reduces to  $e_{\minone}(X; \alpha)$, which depends solely on $X$ and $\alpha$. In addition to propensity scores, we also introduce an outcome working model $Q_{\mintwo}^\dr(X, W, Z; \delta)$ to estimate the conditional expectation $Q_{\minzero}(X, W, Z) = \E(Y \mid X, W, Z)$. Similar to \eqref{eq: dr-X}, the DR estimator for the  set $(X, W)$ can be expressed as follows:
\begin{equation}
    \label{eq: dr-XW}
        \begin{gathered}  
 {\tau}^{\dr}_{\mintwo}( \hat{\alpha}, \hat{\gamma}, \hat{\delta}  ) =  
 \mathbb{P}_n  \begin{bmatrix}
 \dfrac{ZY-\{{Z-e_{\mintwo}(  {{ X}},W;\hat{\alpha},\hat{\gamma})}\}Q^{\dr}_{\mintwo}( {{ X}},W,1;\hat{\delta})}{ e_{\mintwo}(X, W;\hat{\alpha},\hat{\gamma})}   \\\addlinespace[1.5mm] - \dfrac{(1- Z)Y-\{{e_{\mintwo}(X, W;\hat{\alpha},\hat{\gamma})}-Z\}Q^{\dr}_{\mintwo}( {{ X}},W,0;\hat{\delta})}{ 1 - e_{\mintwo}(X, W;\hat{\alpha},\hat{\gamma})}  
 \end{bmatrix},
\end{gathered} 
\end{equation}
where $\hat\alpha$,   $\hat\gamma$ and    $\hat\delta$  are some estimated parameters.
\subsection{Working model conditions}
In many practical applications, correctly specifying the outcome models can be challenging. Therefore, using the ordinary least squares estimates for ${\beta}$ and $\delta$ in $Q_{\minone}^\dr({X},Z;\beta)$ and $Q_{\mintwo}^\dr(X,W,Z;\delta)$ and subsequently plugging them into the standard  DR  estimator can lead to inefficiency and significant variability \citep{cao2009improving, pan2021improved}. In contrast, the propensity score is often more likely to be correctly specified. This motivates us to develop improved DR estimators with more desirable efficiency properties.   Before presenting our theoretical results, we need to impose certain restrictions on the outcome models $Q_{\minone}^\dr({X},Z;\beta)$ and $Q_{\mintwo}^\dr(X,W,Z;\delta)$. This allows us to compare the asymptotic variances of the candidate estimators \eqref{eq: dr-X} and \eqref{eq: dr-XW}. We formally characterize these model restrictions in  following assumption, which encompasses a wide range of models applicable in practice. \begin{assumption} \label{cond:two-res-aipw}
Assume that one of the following conditions holds:
\begin{itemize}[leftmargin=30pt]
    \item[(i)]   The working model $Q_{{{\minone}}}^{\dr}({ X},Z;{\beta})$ is covered by the     working model  $Q_{\mintwo}^\dr({ X},W,Z;{\delta})$. For example, the following linear models are applicable,  
  \begin{equation}
  \label{eq:dr-linear-model}
      \begin{gathered}
   Q_{\minone}^\dr( X  ,Z;{\beta})={ \beta}_0+{\beta}_1Z+{\beta}_2^\T  X +{\beta}_3^\T Z  X,\\ Q_{\mintwo}^\dr({ X},W,Z;{\delta} )={\delta}_0+{\delta}_1Z+{\delta}_2^\T   X +{\delta}_3^\T Z  X+{\delta}_4^\T{ W} +{\delta}_5^\T Z{ W}.
  \end{gathered}
  \end{equation}
\item[  (ii) ]The working model $Q_{\mintwo}^\dr({ X},W,Z;{\delta} )$ is correctly specified.
\end{itemize}
\end{assumption}

Assumption \ref{cond:two-res-aipw}(i) essentially requires that the outcome model employed for  the set $(X,W)$ should include the model corresponding to $X$. This is a common model specification in practice and encompasses a broad class of generalized linear models. For instance, for continuous outcome, we can use linear models or log-linear models, and for binary outcome, logistic regression or probit regression can be applied.  
We do not require the working models  in Assumption \ref{cond:two-res-aipw}(i) to be correctly specified. Additionally, Assumption \ref{cond:two-res-aipw}(ii) complements another possible setup for working models used in semiparametric theory, where   the outcome working model is correctly specified. However, in our context, we only require the correct specification of the $Q_{\mintwo}^\dr( X  ,W,Z;\delta)$ model, without the necessity of correct specification for the $Q_{\minone}^\dr({ X},Z;\beta)$ model. This implies that, in the subsequent sections, we can achieve certain expected efficiency comparison properties as long as Assumption \ref{cond:two-res-aipw}(ii) holds.

\section{Improved DR estimators with  fully fixed propensity score} 
 \label{sec:tps}
\subsection{General strategies} 

In this section, we explore the candidate DR estimators \eqref{eq: dr-X} and \eqref{eq: dr-XW} with fully  fixed propensity score models, represented as $e_{\minone}(X; \alpha^{\ast})$ and $e_{\mintwo}(X, W; \alpha^{\ast}, \gamma^{\ast})$, respectively.  These models do not involve any nuisance parameters \citep{cao2009improving,pan2021improved}. Additionally, the fully fixed parameters $\alpha^{\ast}$ and $\gamma^{\ast}$ may or may not be the same as the true parameters. When $(\alpha^{\ast}, \gamma^{\ast})$ takes its true values $(\alpha_{\minzero}, 0)$, we have $e_{\minone}(X; \alpha^{\ast}) = e_{\mintwo}(X, W; \alpha^{\ast}, \gamma^{\ast})$. For simplicity, we use ${\tau}_{\minone}^{\dr}(\alpha^{\ast}, \hat{\beta})$ and ${\tau}^{\dr}_{\mintwo}(\alpha^{\ast}, \gamma^{\ast}, \hat{\delta})$ for the   DR  estimators \eqref{eq: dr-X} and \eqref{eq: dr-XW} in this section, respectively. It is emphasized that these estimators only rely on the nuisance parameters $ \hat{\beta}$ and $ \hat{\delta}$ of the outcome models.  
In the following discussion, we will only present the  theoretical results for  ${\tau}_{\minone}^{\dr}(\alpha^{\ast}, \hat{\beta})$, with similar conclusions for  ${\tau}^{\dr}_{\mintwo}(\alpha^{\ast}, \gamma^{\ast}, \hat{\delta})$.

 \begin{lemma}
 \label{lem: if}
     Let $\hat{{\beta}}$ be any root-$n$ consistent estimator converging in probability to some ${\beta}^\ast$, that is, $\hat{{\beta}}-{\beta}^\ast=O_p(n^{-1/2})$. When the propensity score is fully known, i.e., ${\alpha}^{\ast} ={\alpha}_{\minzero}$, but $Q^\dr_{\minone}({ X}, Z;{\beta})$ may or may not be,  
  the influence function for ${ {\tau} ^{\dr}_{\minone}( \alpha_\minzero,\hat{\beta} )}$ is $\varphi^\dr_{\minone}(Y, { X},Z;{ {\alpha}}_\minzero, {{\beta}}^*)$, where $
\varphi^\dr_{\minone}(Y, { X},Z;{\alpha} , {{\beta}})$ is denoted as
 \begin{equation}
     \label{eq:eif-dr}\frac{ZY}{e_{\minone} ({ X};\alpha)}-\frac{(1- Z)Y}{ 1-e_{\minone} ({ X};\alpha)}-   \frac{Z-e_{\minone} ({ X};\alpha)}{e_{\minone} ({ X};\alpha)}   Q_{\minone}^\dr({ X},1; {{\beta}})+\frac{e_{\minone} ({ X};\alpha)-Z}{ 1-e_{\minone} ({ X};\alpha)}   Q_{\minone}^\dr({ X},0; {{\beta}})-\tau . 
 \end{equation} 
 \end{lemma}  
 
A proof is provided in Section \ref{sec:proof-tps-if} of supplementary material. The previous results demonstrate that when the propensity score is fully known, i.e., $\alpha^\ast=\alpha_\minzero$, the asymptotic variance of ${\tau}_{\minone}^{\dr}({\alpha}_{\minzero}, \hat{\beta})$ does not depend on the sampling variation of $\hat{{{\beta}}}$ but solely on its limit in probability ${{\beta}}^*$. Utilizing Lemma \ref{lem: if}, we can compute the asymptotic variance of ${\tau}^{\dr}_{\minone}({\alpha}_{\minzero}, \hat{\beta})$ and denote it as $\Sigma_{\minone}^{\drt}( {\beta}^\ast)$.  We use `$\mathrm{fps}$' to characterize the asymptotic variance under the  fully fixed propensity score.  As shown in Section  \ref{sec: asyvar-tps2} of supplementary material, the terms only   involving ${\beta}^\ast$ in  $\Sigma_{\minone}^{\drt}(  {\beta}^*)$ can be expressed as follows: 
\begin{equation*} 
     \begin{aligned} 
   \E \left[  \sqrt{\frac{ 1-e_{\minzero}({ X})}{e_{\minzero}({ X})}}\{Q_{\minzero}({ X},1)-Q_{\minone}^{\dr}({ X},1; {{\beta}}^*)\}+ \sqrt{\frac{e_{\minzero}({ X})}{ 1-e_{\minzero}({ X})}}\{Q_{\minzero}({ X},0)-Q_{\minone}^{\dr}({ X},0; {{\beta}}^*)\} \right]^2.     \end{aligned}  
\end{equation*}  
Let ${\beta}^{\drt}$ be the minimizer of the asymptotic variance $ \Sigma_{\minone} ^{\drt}( {\beta} )$. Taking the derivative of the above equation with respect to ${\beta}$ and setting it to zero, ${\beta}^{\drt}$ is the solution to: 
\begin{equation}
        \label{eq: dr-opt-est}   
\begin{aligned}
 \E&\begin{pmatrix}
\begin{bmatrix}
      \sqrt{\dfrac{1-e_{\minzero}({ X})}{e_{\minzero}({ X})}}\left\{Q_{\minzero}({ X}, 1)-Q_{\minone}^{\dr}({ X}, 1 ; {\beta})\right\} +\sqrt{\dfrac{e_{\minzero}({ X})}{1-e_{\minzero}({ X})}}\left\{Q_{\minzero}({ X}, 0)-Q_{\minone}^{\dr}({ X}, 0 ; {\beta})\right\}
\end{bmatrix} \\\addlinespace[2mm]
\times\left\{\sqrt{\dfrac{1-e_{\minzero}({ X})}{e_{\minzero}({ X})}} Q_{\minone\beta}^{\dr}({ X}, 1 ; {\beta})+\sqrt{\dfrac{e_{\minzero}({ X})}{1-e_{\minzero}({ X})}} Q_{\minone\beta}^{\dr}({ X}, 0 ; {\beta})\right\}
 \end{pmatrix}=0,
\end{aligned}  
\end{equation}
  where $Q_{\minone\beta}^\dr({ X}, Z; {\beta})=\partial Q_{\minone}^\dr({ X},Z;\beta)/\partial\beta$.  
  Hence, if the outcome model $Q^{\dr}_{\minone}({ X},Z ; {{\beta}})$ is correctly specified; that is, $Q^{\dr}_{\minone}({ X},Z ; {{\beta}}_0)=Q_{\minzero}^\dr({ X},Z)$ for some ${{\beta}}_0$, then in fact ${{\beta}}^{\drt}= {{\beta}}_0$. 
If the outcome model is misspecified but the propensity score is fully known, ${{\beta}}^{\drt}$ still exists and minimizes $ \Sigma_{\minone} ^{\drt}( {\beta} )$. However, the ordinary least square estimate ${{\hat{\beta}}}^{\ls}$ solving the estimating equation $ \mathbb{P}_n\left[\{Y-Q^{\dr}_{\minone}({ X},Z ; {{\beta}})\} Q^{\dr}_{{{\minone\beta}}}({ X},Z ; {{\beta}})\right]=0$  does not converge to   ${{\beta}}^{\drt}$. 
Therefore, subsequently plugging ${{\beta}}^{\ls}$ into  \eqref{eq: dr-X}   also does not lead to the smallest asymptotic variance. This explains why the standard DR estimator  is suboptimal when the outcome  model is misspecified \citep{cao2009improving,pan2021improved}. 
In practice, we   consider $\hat{{{\beta}}}^{\drt}$ as the solution to the following   estimating equation,
    \begin{gather} 
    \label{eq:fps-ee-exp}
\mathbb{P}_n 
 \begin{pmatrix}
\begin{bmatrix}
     \dfrac{Z}{e_{\minone}( {{ X}};{\alpha}^{\ast})} \sqrt{\dfrac{1-e_{\minone}( {{ X}};{\alpha}^{\ast})}{e_{\minone}( {{ X}};{\alpha}^{\ast})}}\left\{Y-Q_{\minone}^{\dr}({ X}, 1 ; {\beta})\right\}\\\addlinespace[1.5mm]+\dfrac{1-Z}{1-e_{\minone}( {{ X}};{\alpha}^{\ast})}\sqrt{\dfrac{e_{\minone}( {{ X}};{\alpha}^{\ast})}{1-e_{\minone}( {{ X}};{\alpha}^{\ast})}}\left\{Y-Q_{\minone}^{\dr}({ X}, 0 ; {\beta})\right\}
\end{bmatrix} \\\addlinespace[2mm]
\times\left\{\sqrt{\dfrac{1-e_{\minone}( {{ X}};{\alpha}^{\ast})}{e_{\minone}( {{ X}};{\alpha}^{\ast})}} Q_{\minone\beta}^{\dr}({ X}, 1 ; {\beta})+\sqrt{\dfrac{e_{\minone}( {{ X}};{\alpha}^{\ast})}{1-e_{\minone}( {{ X}};{\alpha}^{\ast})}} Q_{\minone\beta}^{\dr}({ X}, 0 ; {\beta})\right\}
 \end{pmatrix}     =0.   
    \end{gather}  
    When the propensity score is fully known, but the outcome model may not be, the left-hand side of \eqref{eq:fps-ee-exp} converges in probability to the left-hand side of \eqref{eq: dr-opt-est}, hence, $\hat{{{\beta}}}^{\drt} \stackrel{p}{\rightarrow} {{\beta}}^{\drt}$. On the other hand, even when the outcome model is correctly specified but the propensity score may not be, equation \eqref{eq: dr-opt-est} still ensures that $\hat{{{\beta}}}^{\drt}$ remains a consistent estimator for ${{\beta}}_0$. 
The   lemma in section \ref{lem:dr-drest} of supplementary  material   formally establishes the double robustness property of the proposed estimator ${\tau^{\dr}_{\minone}(\alpha^*, \hat{\beta}^{\drt})}$.   
We now present the asymptotic normality of the proposed estimators  using  $M$-estimation theory \citep{stefanski2002calculus}.  We do not discuss the situation when both   propensity score and outcome model  are misspecified, as in such cases, the resulting estimator may not guarantee the consistency.
 \begin{theorem} 
    \label{thm:tps} Given Assumptions \ref{assump: positivity} and \ref{assump: negative-control}, when either the propensity score is fully known or the outcome model is correctly specified, we have
   $  \sqrt{n}\{{ {{\tau}}_{\minone} ^\dr( {\alpha} ^{\ast},\hat{\beta}^\drt  )}-\tau\} \stackrel{d}{\rightarrow} N(0,{{\Lambda}_{\minone} ^\drt} )  $, where the detailed expressions of   $ {{\Lambda}_{\minone} ^\drt}$ can be found in Section  \ref{sec:avar-tps} of supplementary material. 
\end{theorem}

  When the propensity score is fully known, it can be demonstrated that ${\Lambda_{\minone}^\drt} = \Sigma_{\minone}^\drt(\beta^\drt)$, with $\Sigma_{\minone}^\drt(\beta)$ defined earlier above \eqref{eq: dr-opt-est}. In such a scenario,  the proposed estimator would also minimize the asymptotic variance $\Sigma_{\minone}^\drt(\beta)$ of the candidate estimators outlined in \eqref{eq: dr-X}. When the outcome model is correct but the propensity score is incorrect, Theorem \ref{thm:tps} implies that the improved DR estimators remain consistent and asymptotically normal.  

\subsection{Comparative analysis when the propensity score is fully known}
\label{ssec:com-an-fix}
In this section, we will explore whether incorporating   predictive covariates contributes to improving  efficiency when the propensity score is fully known,  namely $e_\minzero(X)=e_{\minone}(X; \alpha^\ast) = e_{\mintwo}(X, W; \alpha^\ast, \gamma^\ast)  $.  In parallel with $\Sigma_{\minone}^\drt(\beta^\drt)$, we introduce $\Sigma_{\mintwo}^\drt(\delta^\drt)$ to represent the asymptotic variance of the influence function defined based on $(X, W)$. The following theorem shows the variance comparison of ${\tau}_{\minone}^{\dr}(\alpha^{\ast}, \hat{\beta})$ and ${\tau}^{\dr}_{\mintwo}(\alpha^{\ast}, \gamma^{\ast}, \hat{\delta})$:
 \begin{theorem}
\label{thm:com-drt}
Assuming that Assumptions \ref{assump: positivity}-\ref{cond:two-res-aipw}  hold,  when the propensity scores     $e_\minone(X;\alpha^\ast)$ and $e_\mintwo(X,W;\alpha^\ast,\gamma^\ast)$  are fully known,  we have $ 
  \Sigma_{\mintwo} ^\drt (   { { {\delta}} ^\drt})\leq  \Sigma_{\minone} ^\drt(   { {\beta}} ^\drt )$. 
\end{theorem}

Theorem \ref{thm:com-drt}  indicates that, when the propensity scores are fully known, incorporating predictive covariates ${W}$ into the improved  DR  estimator further enhances estimation efficiency.  \cite{craycroft2020propensity} has shown that using true confounders and outcome predictors leads to a smaller semiparametric efficiency bound compared to considering only the true confounders. However, to achieve the semiparametric efficiency bound, usual DR estimators typically require correct specifications for both working models \citep{Scharfstein1999JASA, bang2005doubly, kang2007demystifying} or the application of nonparametric estimation methods \citep{hirano2003efficient,chernozhukov2018double}. Theorem \ref{thm:com-drt} implies that our proposed improved DR estimators can achieve similar conclusions even when the outcome models are misspecified. In cases where the propensity score is completely wrong, as discussed earlier, the asymptotic variances   are no longer equivalent to $\Sigma_{\mintwo}^\drt(\delta^\drt)$ and $\Sigma_{\minone}^\drt(\beta^\drt)$, making it infeasible to determine explicit  orders of asymptotic variances.


\section{Improved DR estimators with estimated propensity score}
\label{sec:com-analy}
\subsection{General strategies}
\label{sec:com-analy-est}
 \label{sec:fixed-ps} 
 \label{sec:eps}
  In observational studies, the true propensity score is often unknown and needs to be estimated. To estimate it, a parametric model $e_{\minone}(  { X};{\alpha} )$ is assumed, which involves the  nuisance parameter ${\alpha} $.  
Let $\hat{{\alpha}} $ be the maximum likelihood estimator of ${\alpha}  $ based on $\{(Z_i,  { X}_i)\}_{i=1}^n$.   To begin, we need to write the log-likelihood function for the propensity score model,
\begin{align*} \log f ( { { X}},Z;{ {\alpha}} )& =Z\log\left\{e_{\minone}({ { X}};{ {\alpha}} )\right\}+(1-Z)\log\left\{1-e_{\minone}{({ { X}};{ {\alpha}} )}\right\},
\end{align*}
where $f ( { { X}},Z;{ {\alpha}} )$ is the observed likelihood with respect to  $Z$ and $ { X}$.  The estimate  $ \hat{ {\alpha}}  $ that maximizes the log-likelihood is solution to  the estimating  equations,
\begin{align} 
\label{eq:scorefunction}
{\sum}_{i=1}^n  S_{\alpha}({ X}_{ i
},{Z}_i  ; {{\alpha}}  ) =0 ,   
\end{align}
 where $S_{\alpha}(X, Z; \alpha) = \partial \log f(X, Z; \alpha)/\partial \alpha$ represents the score function. 
Equation \eqref{eq:scorefunction} is sufficient to characterize the estimation of the propensity score model $e_{\minone}(X; \alpha)$. Similarly, we can also consider   maximum likelihood estimation for the propensity score model $e_{\mintwo}(X, W; \alpha, \gamma)$. In addition to replacing $e_{\minone}(X; \alpha)$ in \eqref{eq:scorefunction} with $e_{\mintwo}(X, W; \alpha, \gamma)$, we also need to account for the effect of including additional covariates $W$ in the propensity score model $e_{\mintwo}(X, W; \alpha,\gamma)$. Therefore, we include an additional estimating  equation: 
\begin{align}
\label{eq:ps-est-W}
{\sum}_{i=1}^n  S_{\gamma}({ X}_{ i
},W_i,{Z}_i  ; {{\alpha}}  , {{\gamma}} ) =0 ,  
\end{align}
where $S_{\gamma}(X, W, Z; \alpha, \gamma)$ similarly denotes the score function with respect to the parameter $\gamma$ for the propensity score model $e_{\mintwo}(X, W; \alpha, \gamma)$.  In the following discussion, we will only present the  theoretical results for  ${\tau}_{\minone}^{\dr}( \hat{\alpha}  ,\hat{\beta} )$, with similar conclusions for  ${\tau}^{\dr}_{\mintwo}(\hat {\alpha}  ,\hat\gamma ,\hat{\delta} )$.


\begin{lemma}
\label{lem:if-eps}
      Let $\hat{{\beta}}$ be any root-$n$ consistent estimator converging in probability to some ${\beta}^\ast$, that is, $\hat{{\beta}}-{\beta}^\ast=O_p(n^{-1/2})$. When the propensity score is correctly specified,  but $Q^\dr_{\minone}({ X}, Z;{\beta})$ may or may not be,  the influence function  corresponding to the estimator $ \tau_{\minone}^\dr(\hat\alpha,\hat\beta)$ can be expressed as follows:
 \begin{equation}
     \label{if:dr-est-3}  
\tilde\varphi^\dr_{\minone}(Y, { X},Z;{\alpha}_{\minzero} , {{\beta}}^*)=\varphi^\dr_{\minone}(Y, { X},Z;{ {\alpha}}_{\minzero}, {{\beta}}^*)-{{\mathcal{D}}}_{\minone}^\dr( {{\beta}}^*) {\mathcal{H}}_{{\alpha} {\alpha},  \minzero}^{-1} S_{\alpha}({ X}, Z ;{ {\alpha}}_{\minzero}),
 \end{equation} where $\varphi^\dr_{\minone}(Y, { X},Z;{ {\alpha}}_{\minzero}, {{\beta}} )$ is defined earlier in \eqref{eq:eif-dr},  $
    {{\mathcal{D}}}_{\minone}^\dr({\beta})= -
\E\left\{\partial \varphi^\dr_{\minone}(Y, { X},Z;{ {\alpha}}_{\minzero}, {\beta}) / \partial {\alpha}^\T\right\}  $,  and ${\mathcal{H}}_{{\alpha} {\alpha},  \minzero}=\E\left\{S_{\alpha}({ X}, Z;{ {\alpha}}_{\minzero}) S^\T_{\alpha}({ X}, Z;{ {\alpha}}_{\minzero})\right\}$.  
\end{lemma}

Compared with Lemma \ref{lem: if}, the influence function  \eqref{if:dr-est-3} involves an additional term due to the estimation of the nuisance parameter ${\alpha}$. However, this additional term disappears when both models are correctly specified. We can calculate the asymptotic variance of the influence function \eqref{if:dr-est-3}, denoted as $\Sigma_{\minone}^\dre( \beta)$ for any fixed $\beta$. Here we use `$\mathrm{eps}$'  to represent the asymptotic variance obtained under the estimated propensity score. Define $ {\mathcal{D}}_{\minone\beta}^\dr({\beta})=    \partial{\mathcal{D}}_{\minone}^\dr({\beta})/ \partial {\beta}$, $e_{\minone\alpha}({ X};{\alpha} )=    \partial e_{\minone}({ X};{\alpha} )/ \partial {\alpha}$
and $\tilde e_{\minzero}(X,Z)=Ze_{\minzero}(X)+(1-Z)\{1-e_{\minzero}(X)\}$. The minimizer of the asymptotic variances $\Sigma_{\minone}^\dre( \beta)$ is denoted as ${\beta}^{\dre}$, which is the solution to:
    \begin{align}
    \label{eq: if-est-dr}
    \E&  \begin{pmatrix}  \dfrac{(-1)^{1-Z}}{\tilde e_{\minzero}(X,Z)}  \left\{Y- Q_{\minone}^\dr({ X},Z; {\beta})-{\mathcal{D}}_{\minone}^\dr( {\beta}) \mathcal{H}_{{\alpha} {\alpha},  \minzero}^{-1} e_{\minone\alpha}({ X};{\alpha}_{\minzero}) \right\}
  \\\addlinespace[1mm] \times \begin{bmatrix}
        \dfrac{Z-e_{\minzero}({ X}) }{e_{\minzero}({ X})}  \left\{ Q^\dr_{{\minone\beta}}({ X},1; {{\beta}})+{\mathcal{D}}_{\minone\beta}^\dr( {{\beta}}) {\mathcal{H}}_{{\alpha} {\alpha},  \minzero}^{-1} e_{\minone\alpha}({ X};{\alpha}_{\minzero}) \right\}\\\addlinespace[1mm]-\dfrac{ e_{\minzero}({ X})-Z }{ 1-e_{\minzero}({ X})}   \left\{ Q^\dr_{{\minone\beta}}({ X},0; {{\beta}})+{\mathcal{D}}_{\minone\beta}^\dr( {{\beta}}) {\mathcal{H}}_{{\alpha} {\alpha},  \minzero}^{-1} e_{\minone\alpha}({ X};{\alpha}_{\minzero}) \right\}
    \end{bmatrix} 
    \end{pmatrix}   =    0.
\end{align} Detailed derivations of    $\Sigma_{\minone}^\dre( \beta)$  are relegated to Section \ref{sec:der-est-expression} of supplementary material.  
Let $\hat{\mathcal{H}}_{{\alpha} {\alpha}}=\mathbb{P}_n\left\{S_{\alpha}({ X},Z; \hat{{\alpha}}) S^\T_{\alpha}({ X},Z; \hat{{\alpha}})\right\}$, $ \hat{{\mathcal{D}}}^\dr_{\minone}({\beta})=-\mathbb{P}_n\left\{\partial  {\varphi}^\dr_\minone(Y, X,Z; \hat{{\alpha}}, {\beta}) / \partial {\alpha}^\T\right\}$,   ${\hat{{\mathcal{D}}}} _{\minone\beta}^\dr({\beta})=-\mathbb{P}_n\left\{\partial^2  {\varphi}_\minone^\dr(Y,{ X}, Z; \hat{{\alpha}},{\beta}) / \partial {\alpha}^\T \partial {\beta}\right\}$, and $e_{\minone} ({ X},Z;\hat{\alpha})=Ze_{\minone} (X;\hat\alpha)+(1-Z)\{1-e_{\minone} (X;\hat\alpha)\}$.  We hence consider   $\hat{\beta}^{\dre}$ as the solution to the sample version of \eqref{eq: if-est-dr},   
\begin{equation}
    \label{eq:solve-est-eq}
        \begin{aligned}
   \mathbb{P}_n   \begin{pmatrix} \dfrac{(-1)^{1-Z}}{\tilde e (X,Z;\hat\alpha)}  \left\{Y- Q_{\minone}^\dr({ X},Z; {\beta})-\hat{\mathcal{D}}_{\minone}^\dr( {\beta}) \hat{\mathcal{H}}_{{\alpha} {\alpha} }^{-1} e_{\minone\alpha}({ X};{\hat\alpha}) \right\}
  \\\addlinespace[1mm] \times \begin{bmatrix}
        \dfrac{Z- e_{\minone} ({ X};{\hat\alpha})}{ e_{\minone} ({ X};{\hat\alpha})}  \left\{ Q^\dr_{{\minone\beta}}({ X},1; {{\beta}})+\hat{\mathcal{D}}_{\minone\beta}^\dr( {{\beta}}) \hat{\mathcal{H}}_{{\alpha} {\alpha},  \minzero}^{-1} e_{\minone\alpha}({ X};{\hat\alpha}) \right\}\\\addlinespace[1mm]-\dfrac{  e_{\minone} ({ X};{\hat\alpha})-Z }{ 1- e_{\minone} ({ X};{\hat\alpha})}   \left\{ Q^\dr_{{\minone\beta}}({ X},0; {{\beta}})+\hat{\mathcal{D}}_{\minone\beta}^\dr( {{\beta}}) \hat{\mathcal{H}}_{{\alpha} {\alpha},  \minzero}^{-1} e_{\minone\alpha}({ X};{\hat\alpha}) \right\}
    \end{bmatrix} 
    \end{pmatrix} =0 ,
\end{aligned} 
\end{equation} where $\tilde e_{\minone} (X,Z;\hat\alpha)=Z  e_{\minone} (X;\hat\alpha)+(1-Z)\{1-e_{\minone} (X;\hat\alpha)\}$.  When the propensity score is correct, $\hat{\beta}^\dre\stackrel{p}{\rightarrow} {\beta}^\dre$; {and when the outcome model is correct, $\hat{\beta}^\dre \stackrel{p}{\rightarrow} {{\beta}}_0$.}  Therefore, the improved DR estimator ${\tau}^\dr_\minone(\hat\alpha,\hat\beta^\dre)$  maintains its double robustness property under the estimated propensity score.   When both models are correctly specified, using the estimated propensity score can also achieve the semiparametric efficiency bound, as discussed in Section \ref{sec:dr-ipw-eps}  of supplementary material. The following theorem establishes that the improved DR estimator is   
 asymptotically normal  when either nuisance model is correctly specified.
\begin{theorem} 
\label{thm:eps}  Given Assumptions \ref{assump: positivity} and \ref{assump: negative-control}, 
  when either the propensity score or the outcome model is correctly specified, we have $
\sqrt{n} \{\tau^\dr_\minone(\hat{\alpha},\hat{\beta}^\dre)-\tau \} \stackrel{d}{\rightarrow} N(0, \Lambda^\dre_\minone )    $, where the detailed expressions of  $ {{\Lambda}_{\minone} ^\dre}$ can be found in  Section \ref{sec:prof-thm3} of supplementary material.
\end{theorem} 

   When the propensity score model is correctly specified, we know that  ${{\beta}}^{\dre}$  is the minimizer of the asymptotic variance of   \eqref{if:dr-est-3} and ${{\Lambda}_{\minone} ^\dre}  =\Sigma_{\minone}^\dre(\beta^\dre)$. In such  scenario, for the  potentially misspecified outcome working models, we have:
\begin{equation}
\label{eq:ipw-var-ps}
\begin{aligned} 
  \Sigma_{\minone}^{\dre} ({\beta})& =\Sigma_{\minone}^{\drt} ({\beta})  -{{\mathcal{D}}}_{\minone}^\dr( {{\beta}} ) {\mathcal{H}}_{{\alpha} {\alpha},  \minzero}^{-1} {{\mathcal{D}}}_{\minone}^\T( {{\beta}} )  
\end{aligned}
\end{equation}
for any $\beta$. 
The proof can be found in Section \ref{sec: proof-of-eq10} of supplementary material. 
This finding aligns with widely accepted results in causal inference, indicating that even when the propensity score is known, using the estimated propensity score can improve estimation efficiency \citep{hirano2003efficient}. As demonstrated earlier, when  both nuisance models are correctly specified, the additional term in \eqref{eq:ipw-var-ps} would disappear. Therefore, in such a case, we have $\Sigma_{\minone}^{\dre} ({\beta}^{\dre}) = \Sigma_{\minone}^{\drt} ({\beta}^{\drt})$, and the improved DR estimator would also achieve semiparametric efficiency bound as the standard DR estimator.

\subsection{Comparative analysis when the propensity score is correctly specified}
\label{ssec:comp-est-ps}

In this section, we will explore whether integrating predictive covariates contributes
to improving eﬀiciency when the propensity score is  correctly specified, based on the theoretical foundation established in Section \ref{sec:com-analy-est}.  Intuitively, when the propensity scores are fully known, the efficiency improvement is attributed to the introduction of additional predictive covariates in   outcome models through Assumption \ref{cond:two-res-aipw}. However, when the propensity score is correctly specified but requires estimation, even if $W$ has no direct effect on the treatment $Z$, and the probability limit of $\hat{\gamma}$ is $\gamma_{\minzero}=0$, estimating ${\tau}^{\dr}_{\mintwo}(\hat {\alpha} ,\hat\gamma ,\hat{\delta} )$ still necessitates an additional estimating equation for the parameter ${\gamma}$, as indicated in   \eqref{eq:ps-est-W}. 
Therefore, compared to Section \ref{ssec:com-an-fix}, in addition to exploring efficiency through the outcome model, we also need to consider whether using estimated propensity scores could potentially lead to efficiency gains, achieving a dual improvement through two nuisance models. 

As previously described, by introducing the additional estimating equation \eqref{eq:ps-est-W}, we can also calculate the asymptotic variance of ${\tau}^{\dr}_{\mintwo}(\hat {\alpha} ,\hat\gamma ,\hat{\delta} )$ based on the set $({X},{W})$ when the propensity score model $e_{\mintwo}(X, W;{\alpha},{{{\gamma}}})$ is correctly specified. To simplify the exposition, we review key definitions related to asymptotic variances. Specifically, we use $\Sigma_{\minone}^{\dre} ( {\beta}  )$ to represent the asymptotic variance under the correctly specified propensity score, corresponding to the covariate set $X$ for the parameter vector ${\beta}$.  Likewise, $\Sigma_{\mintwo}^{\dre} ( {\delta}  )$ represents the asymptotic variance  for the covariate set $(X, W)$ corresponding to the parameter vector ${\delta}$. We use  $\varphi^\dr_{\mintwo}(Y, X, W, Z; \alpha_{\minzero}, \gamma_{\minzero}, \delta)$ to  denote the influence function for ${\tau}^{\dr}_{\mintwo}(\alpha\minzero, \gamma_\minzero, \hat{\delta})$ with $\hat{\delta} - \delta = O_p(n^{-1/2})$.  The Fisher information matrix  ${\mathcal{H}}_{{\alpha} {\alpha},  \minzero}$ is denoted as ${\mathcal{H}}_{{\alpha} {\alpha},  \minzero}=\E\left\{S_{\alpha}({ X}, Z;{ {\alpha}}_{\minzero} ) S^\T_{\alpha}({ X} , Z;{ {\alpha}}_{\minzero} )\right\} $, which also equals $\E\left\{S_{\alpha}({ X},{ W}, Z;{ {\alpha}}_{\minzero}, \gamma_{\minzero}) S^\T_{\alpha}({ X},{ W}, Z;{ {\alpha}}_{\minzero}, \gamma_{\minzero})\right\}  $ due to $e(X;\alpha_\minzero)=e(X,W;\alpha_\minzero,\gamma_\minzero)$. Let  ${\mathcal{H}}_{{\alpha} {\gamma},  \minzero}=\E\left\{S_{\alpha}({ X},{ W}, Z;{ {\alpha}}_{\minzero}, \gamma_{\minzero}) S^\T_{\gamma}({ X},{ W}, Z;{ {\alpha}}_{\minzero}, \gamma_{\minzero})\right\} $, and the matrices       ${\mathcal{H}}_{{\gamma} {\gamma},  \minzero}$ and ${\mathcal{H}}_{{\gamma} {\alpha},  \minzero}$  can be similarly defined.  Let $\mathcal{F} =  \mathcal{H}_{{\alpha}  {\alpha},  \minzero}-\mathcal{H}_{{\alpha}  \gamma, \minzero} \mathcal{H}_{\gamma \gamma, \minzero}^{-1} \mathcal{H}_{\gamma  {\alpha},  \minzero}  $, and it can be shown that $\mathcal{F}^{-1}$ is a positive definite matrix;  refer to Section \ref{eq:verf-var-comp} in the supplementary material for details.   We also define $   {\mathcal{D}}_{\mintwoone}^\dr( {\delta} )  = -
\E\left\{\partial \varphi^\dr_{\mintwo} (Y, { X},{ W},Z ; {\alpha} _{\minzero}, \gamma_{\minzero}, {\delta} ) / \partial {\alpha}^\T\right\} $, $    {\mathcal{D}}_{\mintwotwo}^\dr( {\delta} )  = -
\E\left\{\partial\varphi^\dr_{\mintwo} (Y, { X},{ W},Z ; {\alpha} _{\minzero}, \gamma_{\minzero}, {\delta} ) / \partial {\gamma}^\T\right\}  $ and \begin{gather*} 
\mathcal{M}^\dr (\delta )= \{ {\mathcal{D}}_{\mintwotwo}^\dr(\delta )-{\mathcal{D}}_{\mintwoone}^\dr(\delta )   \mathcal{H}_{{\alpha}  {\alpha}, \minzero}^{-1} \mathcal{H}_{{\alpha} \gamma, \minzero}\}\mathcal{F}^{-1}\{ {\mathcal{D}}_{\mintwotwo}^\dr(\delta )-{\mathcal{D}}_{\mintwoone}^\dr(\delta )   \mathcal{H}_{{\alpha}  {\alpha},  \minzero}^{-1} \mathcal{H}_{{\alpha} \gamma, \minzero}\}^\T .
\end{gather*}
 For any $\delta$, we have $\mathcal{M}^\dr (\delta )\geq 0$  due to the positive definiteness of matrix $\mathcal{F}^{-1}$.  Through elaborate computations in Section \ref{eq:verf-var-comp} of the supplementary material,   we have
     \begin{align} 
      \label{est:var-comp2}
       \Sigma_{\mintwo} ^{\dre}(  {\delta} )    &=\Sigma_{\mintwo}^{\drt}  (   {\delta} ) -{\mathcal{D}}_{\mintwoone}^\dr( {\delta} )   \mathcal{H}_{{\alpha}  {\alpha},  \minzero}^{-1}  {\mathcal{D}}_{\mintwoone}^{\T}(  {\delta} ) -\mathcal{M}^\dr({\delta} ).
\end{align}   
The derived equation \eqref{est:var-comp2} is particularly useful. Firstly, similar to \eqref{eq:ipw-var-ps}, it also indicates that using estimated propensity scores can enhance estimation efficiency. Secondly, it can be used to compare the  asymptotic variances under two different adjustment sets using estimated propensity scores, even in situations where the outcome working models may be misspecified.  
Section 3.3 of \citet{lunceford2004stratification}   investigates a similar expression when introducing additional predictive covariates $W$ into the  IPW  estimator. However, for  DR type  estimators, this  paper does not discuss whether the use of predictive covariates under a potentially misspecified outcome model could also improve efficiency. We next illustrate how \eqref{est:var-comp2} is applied to compare asymptotic variances of the improved DR estimators.


\begin{example}
We employ the linear models \eqref{eq:dr-linear-model} for illustrative purposes. As mentioned earlier, the linear models considered here do not necessitate correct specification. To facilitate the comparison of variances, we extend the parameter vector ${\beta} ^\dre$ to accommodate the model \eqref{eq:dr-linear-model}. Specifically, we extend the vector as $\tilde{\delta} = \left\{({\beta} ^\dre)^\T, 0_{1\times d_w}, 0_{1\times d_w}\right\}^\T$, where $ {  d_w}$ is the dimension of covariates $W$. 
By definition, two estimators $\tau_\minone(\alpha_\minzero,\beta^\dre)$  and $\tau_\mintwo(\alpha_\minzero,\gamma_\minzero,\tilde\delta)$  are identical,   indicating that they should have the same asymptotic variance, i.e., $\Sigma_{\mintwo}^{\drt} (\tilde{\delta}) =\Sigma_{\minone}^{\drt} ( {\beta} ^\dre)$, where $\Sigma_{\minone}^{\drt} (\beta^\dre)$ is the asymptotic variance of $\tau_\minone(\alpha_\minzero,\beta^\dre)$, and $\Sigma_{\mintwo}^{\drt} (\tilde{\delta}) $ is the asymptotic variance of $\tau_\mintwo(\alpha_\minzero,\gamma_\minzero,\tilde{\delta})$. Additionally, by direct verification of these definitions, we also have ${\mathcal{D}}_{ \minone}^\dr( {\beta} ^\dre)  ={\mathcal{D}}_{\mintwoone}^\dr(\tilde{\delta}) $. Plugging $\tilde{\delta}$ into \eqref{est:var-comp2}, we have: 
      \begin{align*}
         \Sigma_{\mintwo} ^{\dre}(  \tilde{\delta})  &=\Sigma_{\mintwo}^{\drt}  (\tilde{\delta}) -{\mathcal{D}}_{\mintwoone}^\dr(\tilde{\delta})   \mathcal{H}_{{\alpha}  {\alpha},  \minzero}^{-1}  {\mathcal{D}}_{\mintwoone}^{\T}( \tilde{\delta} ) -\mathcal{M}^\dr (  \tilde{\delta} )  \\&=\Sigma_{\minone}^{\drt} ( {\beta} ^\dre)  -{\mathcal{D}}_{ \minone}^\dr( {\beta} ^\dre)  \mathcal{H}_{{\alpha}  {\alpha},  \minzero}^{-1}  {\mathcal{D}}_{ \minone}^{\T}( {\beta} ^\dre) -\mathcal{M}^\dr (  \tilde{\delta} )  \\&=\Sigma_{\minone}^{\dre}  ( {\beta} ^\dre)-\mathcal{M}^\dr (  \tilde{\delta} ) ,
 \end{align*}
 where the final equality holds due to \eqref{eq:ipw-var-ps}. Since $\mathcal{M}^\dr (  \tilde{\delta} ) \geq 0$ for any $ \tilde{\delta} $, we have that $$ \Sigma_{\mintwo} ^{\dre}( {\delta} ^\dre)\leq \Sigma_{\mintwo} ^{\dre}( \tilde{\delta} )\leq \Sigma_{\minone} ^{\dre}( {\beta}^\dre ) .$$
\end{example}
The above example illustrates that under the most commonly used linear models, we can always reduce asymptotic variance by incorporating predictive covariates, even if the outcome working models may be misspecified.  In summary, we provide a more general result.

\begin{theorem}
    \label{thm:com-dre}
Assuming that Assumptions \ref{assump: positivity}-\ref{cond:two-res-aipw}  hold,  when the propensity scores $e_\minone(X;\alpha)$ and $e_\mintwo(X,W;\alpha,\gamma)$  are correctly specified,   we have   $$ \Sigma_{\mintwo} ^{\dre}( {\delta} ^\dre)\leq   \Sigma_{\minone} ^{\dre}( {\beta}^\dre )\leq   \Sigma_{\minone} ^{\drt}( {\beta}^\drt ) , ~~~~\Sigma_{\mintwo} ^{\dre}( {\delta} ^\dre)\leq   \Sigma_{\mintwo} ^{\drt}( {\delta}^\drt )\leq   \Sigma_{\minone} ^{\drt}( {\beta}^\drt ) .$$
\end{theorem}

Theorem \ref{thm:com-dre} indicates that the proposed  DR  estimators lead to a dual efficiency improvement when introducing predictive covariates for both   outcome model and   estimated propensity score model. It is noteworthy that we do not need to impose new modeling restrictions on the propensity score in practice, as the correct specification of the   model $e_{\minone}(X; \alpha)$ always ensures the correct specification of $e_{\mintwo}(X, W; \alpha, \gamma)$. 
Hence, practitioners should consistently utilize estimated propensity scores and adjust for   predictive covariates to enhance estimation efficiency. From this perspective, the proposed estimators in this paper are more accurately characterized as model-assisted rather than strictly model-based. Furthermore, when the outcome models are correctly specified, the improved  DR estimator still provides protection for consistent estimation.  When all nuisance models are correctly specified, our proposed estimators can also achieve the semiparametric efficiency bound.  Additionally,   adjusting for covariates $(X, W)$ would result in a smaller semiparametric efficiency bound compared to adjusting for covariates $X$ alone \citep{craycroft2020propensity}. 

\section{Simulation Studies} 
\label{sim-studies}
We conducted several simulation studies to evaluate the finite sample performance of our proposed method. We simulate with a sample size of $n = 1000$ under the following data-generating mechanism:
\begin{itemize}[leftmargin=30pt]
\item[(a)] The true confounders $X=(X_1,X_2)^\T$ are drawn from bivariate normal distributions, with $\E(X_1)=\E(X_2)=0$, $\mathrm{var}(X_1)=1$, $\mathrm{var}(X_2)=2$, and $\mathrm{cov}(X_1,X_2)=1$. 
\item[(b)]The predictive coavariate is generated from $W=0.5X_2^2+\epsilon_w$, where $\epsilon_w\sim N(0,1)$.
    \item[(c)] The treatment variable $Z$ follows a Bernoulli distribution with  $\pr(Z=1\mid  X )={ \mathrm{expit}}(-0.2+X_1-0.2X_2)$, where $\mathrm{expit}(u)=\exp(u)/\{1+\exp(u)\}$.
     \item[(d)]  The observed outcome is generated from $  Y=  2+2Z-X_1+0.5X_2^3 -W + \epsilon_y   $,  
   where $\epsilon_y \sim N(0,1)$. 
\end{itemize}

The true  ATE  is 2. We assess the performance of the proposed estimators under two adjustment sets. Estimating equations \eqref{eq:fps-ee-exp} and \eqref{eq:solve-est-eq} will be used to estimate the parameters. The proposed estimators are compared against the standard  DR estimator \citep{robins1994estimation, bang2005doubly}. In the standard method, we will utilize least squares to solve the outcome model. We consider the following modeling choices for the working models.
 \begin{enumerate}[leftmargin=30pt]
     \item[(1)]   {\texttt{CC}}: correct propensity score and correct outcome model.         \item[(2)]   {\texttt{IC}}:    incorrect propensity score and correct outcome model. In this scenario, both propensity scores are misspecified by assuming that  
  \begin{equation*} 
      \begin{gathered}e_{\minone}( {{ X}};\alpha)=e_{\mintwo}(X, W;\alpha,\gamma)=\mathrm{expit}(\alpha_1+\alpha_2X_1^\dag+\alpha_3X_2^\dag),
  \end{gathered}
  \end{equation*} where $X_1^\dag$ and $X_2^\dag$ are two transformed variables \citep{kang2007demystifying}. Specifically, $X_1^\dag$ follows a Bernoulli distribution with $\pr(X_1^\dag=1)=0.4$, and  $X_2^\dag\sim\text{U}(0,1)$.   
When the propensity score is fully fixed, we adopt the misspecified models $e_{\minone}(X;\alpha^\ast)$ and $ e_{\mintwo}(X, W;\alpha^\ast,\gamma^\ast)$, where $\alpha^\ast=(\alpha_1^\ast, \alpha_2^\ast,\alpha_3^\ast)^\T = (0.1, 0,0.1)^\T$ and $\gamma^\ast=0$. 
When the propensity score is estimated, we use the maximum likelihood estimation.  \item[(3)] {\cmtt {\texttt{CI}}}:  correct propensity score and incorrect outcome model.  In this scenario,  $Q_{\minzero}(Z,X)$ and $Q_{\minzero}(Z,X,W)$ are misspecified and follow the following linear models, respectively:
  \begin{equation*} 
      \begin{gathered}
 Q_{{{\minone}}}^{\dr}({ X},Z;{\beta})={ \beta}_0+{\beta}_1  X ^\circ  ,~ Q_{\mintwo}^\dr({X},{W},Z;{\delta} )={ \delta}_0+ {\delta}_1   X ^\circ   +{\delta}_2{ W^\circ }  ,
  \end{gathered}
  \end{equation*}
  where   $ X^\circ = \vert X_1X_2\vert$ and  $ W^\circ =\mathrm{exp}(W)/\{1+\mathrm{exp}(X_2)\}$.
        \item[(4)]  {\texttt{II}}: incorrect propensity score and incorrect outcome model. In this scenario, the misspecified propensity score model follows the setting    in {\texttt{IC}}, and the misspecified outcome model follows the setting    in {\texttt{CI}}.
 \end{enumerate}
Table  \ref{tab: res2} presents bias,  variance and corresponding 95\% confidence interval based coverage probability (CP)  for the average treatment effect. We calculated the variance using bootstrap procedures with 500 replications.  To summarize the findings from Table \ref{tab: res2}: 
\begin{enumerate}[leftmargin=30pt]
\renewcommand{\labelenumi}{(\arabic{enumi})}

    \item  Verification of the double robustness property. 
For both the standard DR  and improved DR  estimators, our proposed estimators \eqref{eq: dr-X} and \eqref{eq: dr-XW} exhibited stability with relatively small bias when the propensity score or outcome model is correctly specified (see scenarios   {\texttt{IC}},   {\texttt{CI}} and  {\texttt{CC}}). Additionally, the 95\% coverage probability closely aligns with the nominal level in the majority of cases.
 
\begin{table}[t]
\caption{Simulation results. Bias  and   variances are scaled by sample size $n=1000$. The 95\%  coverage probability   is scaled by 100. }
\label{tab: res2}
\begin{threeparttable}  
\centering
\resizebox{0.780250\columnwidth}{!}{%

\begin{tabular}{ccccccccccc} 
\hline
        &           &  &  & \multicolumn{7}{c}{Fully fixed propensity score}                                   \\ 
 \cline{5-11}\addlinespace[1mm]
        &           &  &  & \multicolumn{3}{c}{Usual estimators} &  & \multicolumn{3}{c}{Improved estimators}  \\ 
\cline{5-7}\cline{9-11}\addlinespace[1mm]
Set     & Scenarios &  &  & Bias   & Var   & CP                  &  & Bias    & Var   & CP                     \\  
\cline{1-2}\cline{5-7}\cline{9-11}\addlinespace[1mm]
$X$     &  {\texttt{CC}}        &  &  & 2.07   & 9.87  & 92.04               &  & 3.89    & 10.69 & 95.52                  \\
$X$     &  {\texttt{IC}}        &  &  & 2.39   & 9.48  & 92.04               &  & 7.32    & 53.65 & 100.00                 \\
$X$     &  {\texttt{CI}}        &  &  & 0.51   & 18.15 & 93.53               &  & -1.95   & 17.26 & 92.04                  \\
$X$     &  {\texttt{II}}        &  &  & -89.63 & 10.73 & 45.27               &  & -96.76  & 10.40 & 36.82                  \\
$(X,W)$ &  {\texttt{CC}}        &  &  & 2.58   & 4.88  & 94.53               &  & 4.12    & 5.33  & 97.51                  \\
$(X,W)$ &  {\texttt{IC}}        &  &  & 3.30   & 4.73  & 93.03               &  & 8.37    & 33.73 & 99.00                  \\
$(X,W)$ &  {\texttt{CI}}        &  &  & -4.24  & 17.19 & 93.53               &  & -12.80  & 16.43 & 92.54                  \\
$(X,W)$ &  {\texttt{II}}        &  &  & -97.72 & 12.90 & 41.29               &  & -104.76 & 13.32 & 34.83                  \\ 
\hline
        &           &  &  & \multicolumn{7}{c}{Estimated propensity score}                                   \\ \addlinespace[1mm]
 \cline{5-11}\addlinespace[1mm]
        &           &  &  & \multicolumn{3}{c}{Usual estimators} &  & \multicolumn{3}{c}{Improved estimators}  \\ 
\cline{5-7}\cline{9-11}\addlinespace[1mm]
Set     & Scenarios &  &  & Bias   & Var   & CP                  &  & Bias    & Var   & CP                     \\  
\cline{1-2}\cline{5-7}\cline{9-11}\addlinespace[1mm]
$X$     &  {\texttt{CC}}        &  &  & 2.90   & 10.23 & 92.04               &  & 5.58    & 10.41 & 92.04                  \\
$X$     &  {\texttt{IC}}        &  &  & 0.99   & 9.61  & 92.54               &  & -2.38   & 9.39  & 95.52                  \\
$X$     &  {\texttt{CI}}        &  &  & -9.48  & 13.17 & 95.02               &  & -10.38  & 11.97 & 94.03                  \\
$X$     &  {\texttt{II}}        &  &  & -89.53 & 11.01 & 50.25               &  & -90.46  & 11.01 & 50.25                  \\
$(X,W)$ &  {\texttt{CC}}        &  &  & 2.90   & 5.18  & 94.03               &  & 6.66    & 5.28  & 95.52                  \\
$(X,W)$ &  {\texttt{IC}}        &  &  & 2.23   & 4.77  & 93.53               &  & 0.94    & 4.37  & 95.52                  \\
$(X,W)$ &  {\texttt{CI}}        &  &  & -6.82  & 11.45 & 96.02               &  & -11.12  & 6.99  & 93.03                  \\
$(X,W)$ &  {\texttt{II}}        &  &  & -98.07 & 13.05 & 42.79               &  & -98.32  & 13.42 & 43.28                  \\
\hline
\end{tabular}}
\end{threeparttable}
\end{table} 

\item  
Verification of   variance reduction through the estimated propensity score. Empirical evidence consistently shows that estimators utilizing the estimated propensity score generally exhibit smaller variances compared to the estimators using the known propensity score, particularly in the scenario {\texttt{CI}}.
It is noteworthy that in many scenarios with fixed propensity scores,  the 95\% coverage probability shows poor performance. However, such poor performance is not observed when using  the estimated propensity score. This inconsistent performance arises because the selected parameters $ \alpha^\ast $ and $\gamma^\ast$ for fixed propensity scores result in extreme weights.

\item   Verification of semiparametric eﬀiciency bound.  Under the scenario {\texttt{CC}}, whether using fixed fully known propensity scores or estimated propensity scores, both the standard DR and improved DR estimators achieve the semiparametric efficiency bound and exhibit similar asymptotic performance. Additionally, the semiparametric efficiency bound estimated using $(X,W)$ is even smaller.

          \item   Verification of smallest asymptotic variance.  
In scenarios where the propensity model is correct but the outcome  model is misspecified ({\texttt{CI}}), the improved DR achieves the smallest asymptotic variance among all candidate DR estimators, naturally outperforming the standard DR estimators.

\item Verification of the optimal adjustment set.
The performance of the improved DR estimators, defined based on two adjustment sets, aligns with the findings presented in Theorems \ref{thm:com-drt} and \ref{thm:com-dre}. The improved estimators utilizing the adjustment set $(X,W)$ achieve a lower variance compared to those using only the set $X$, especially in the case of the estimated propensity score. 


\end{enumerate}

\section{Application}
\label{sec:app}


In this section, we apply the proposed improved DR estimators to a dataset from the Study to Understand Prognoses and Preferences for Outcomes and Risks of Treatments (SUPPORT, \citet{connors1996effectiveness}). SUPPORT examined the effectiveness and safety of the direct measurement of cardiac function by Right Heart Catheterization (RHC) for certain critically ill patients in intensive care units (ICUs). This dataset has been previously analyzed to estimate the average treatment effect of using RHC \citep{lin1998assessing, tan2006distributional, cui2021semiparametric}. The dataset includes information on 5735 patients, with 2184 measured by RHC in the first 24 hours ($Z=1$) and 3551 in the control group ($Z=0$). Among these, 3817 patients survived, while 1918 patients passed away within the 30-day period. The outcome $Y$ is the number of days between admission and death or censoring at day 30. Our objective is to evaluate whether the usage of RHC contributes to the survival time of critically ill patients from the day admitted or transferred to the ICU. 
 
Due to the relatively high  missing rates, covariates such as \texttt{cat2} (secondary disease category, 79\% missing, denoted as $ \tilde W_1$),  \texttt{adld3p} (activities of daily living, 75\% missing, denoted as $\tilde W_2$), and \texttt{urin1} (urine output, 53\% missing, denoted as $\tilde W_3$) are typically excluded in the previous  literature \citep{hirano2001estimation}. Consequently, it is widely acknowledged that these covariates do not qualify as genuine confounders, although they may still have predictive power for outcome. As mentioned below Assumption \ref{assump: negative-control}, we can utilize their observed values and impute the missing ones    for the purpose of analysis. 
For binary variables $\tilde W_1$ and $\tilde W_2$, we utilize the Bernoulli distribution for imputation, relying on the observed sample mean. As for the continuous variable $\tilde W_3$, we apply the normal distribution for imputation, considering both the observed sample mean and variance. We refer to the three imputed variables as $W_1$, $W_2$, and $W_3$. 


In our collected covariates, all the remaining variables are complete without any missing values. Additionally, the covariate set includes environmental factors, especially temperature. Temperature variables are not expected to directly influence the choice of patient treatment but may affect the duration of a patient's hospital stay. Therefore, we consider incorporating the temperature variable \texttt{temp1} as $W_4$. The additional predictive covariates in this study are denoted as $W=(W_1, W_2, W_3, W_4)$. We employed a logistic model to fit the propensity score model $e_\minone(X;\alpha)$ and empirically verified that $W$ does not directly appear in the model. Furthermore, we explored alternative non-parametric methods such as classification and regression trees \citep{Breiman1984c}, Bayesian additive regression trees \citep{chipman2010bart}, and Causal Tree \citep{athey2016recursive}. Importantly, almost all of these methods empirically confirmed that the covariate $W$ do not directly affect  treatment assignment. The remaining covariates, constituting the set $X$, include a total of 48 variables. These variables encompass diagnoses, comorbidities, vital signs, physiological indicators, and additional demographic details such as gender, race, age, education, and income. For further details, please refer to Table 1 in \citet{hirano2001estimation}.

\begin{table}[t]
\caption{Treatment effect estimates, asymptotic variance and 95\% confidence intervals  for the RHC data.}
\label{tab: application} 
\centering 
\resizebox{0.975\textwidth}{!}{ 
\begin{tabular}{cclcccclccc}
\hline
              &  &  & \multicolumn{3}{c}{Usual Estimators} &  &  & \multicolumn{3}{c}{Improved Estimators} \\ \cline{4-6} \cline{9-11} \addlinespace[1mm]
Scenarios     &  &  & Point estimate &   Variance  & 95\% CI           &  &  & Point estimate  &   Variance   & 95\% CI            \\ \cline{1-1} \cline{4-6} \cline{9-11} \addlinespace[1mm]
 $X$     &  &   & 3.26 & 0.74 & (1.58,~4.94) &  &   & 3.63 & 0.69 & (2.00,~5.26) \\ 
 $(X,W)$  &  &   & 3.26 & 0.74 & (1.59,~4.94)  &  &  & 3.80 & 0.67 & (2.19,~5.42) \\ 
  \hline
\end{tabular}
}
\end{table}

We employed the linear model \eqref{eq:dr-linear-model} as the outcome model.
Table \ref{tab: application} presents point estimates, variance, and corresponding 95\% confidence intervals for ATE. Since the true propensity score is not known in practice, we thus adopt the methods proposed in Section \ref{sec:com-analy} for   data analysis.  The point estimates of the improved DR estimator are slightly larger than those of the standard DR estimator, but they are still very close. The consistency among these close estimates suggests that  RHC  treatment may increase the survival of critically ill patients by at least an additional 3 days after entering the intensive care unit. These findings align with results reported in earlier studies \citep{Luke2021JASA}. We also observe  that employing the improved DR estimator leads to a reduction in estimated variance compared to the standard DR estimator. Furthermore, through the integration of additional predictive covariates $W$, we find  a further decrease in the estimated variance.


Finally, we would like to emphasize that, for real data analysis, there is no need to partition the covariates $X$ without missing values to define $W$. In fact, we strongly recommend incorporating additional covariates that may be missing in practice, as they are often not considered genuine confounders, such as $(\tilde W_1, \tilde W_2, \tilde W_3)$ in this application. After appropriate imputation, we recommend including them in the propensity score for empirical validation, ensuring that these covariates do not directly impact the treatment variable. Subsequently, these variables can be employed in the outcome model, providing a highly attractive method for further enhancing efficiency in practical applications. 

\section{Discussion}  
\label{sec: dis}

In this paper, we introduce an enhanced doubly robust estimator that directly minimizes the asymptotic variance within a defined class of working models.   Our estimator is doubly robust and is designed to be more efficient than standard DR estimators when the propensity score model is correctly specified. Additionally, we theoretically establish that the inclusion of additional covariates with predictive power   can further enhance estimation efficiency, even in the presence of missing values for such covariates.



For the usual  DR  estimator, it remains uncertain whether similar variance orders exist when the propensity scores are correctly specified. This uncertainty arises because the usual DR estimator typically relies on ordinary least squares regression  and does not incorporate any propensity score information. Our approach is intuitively similar to the covariate adjustment strategies used in randomized experiments.  In such scenarios, although the difference-in-means estimator is sufficient for  consistent estimation, previous literature often suggests incorporating additional predictive covariates to
 improve estimation efficiency.


There are several potential directions for future research.  Firstly, while this paper demonstrates that incorporating predictive covariates with missing values can enhance estimation efficiency, exploring various imputation methods is also worth considering \citep{Zhao2022JASA}.   Additionally, this paper primarily focuses on covariates with predictive power but no direct effect on the treatment. In practice, exploring covariates that affect the treatment but have no direct effect on the outcome is also of significant interest \citep{craycroft2020propensity}. Finally, 
 a promising direction is to explore bias-reduced doubly robust estimator  for comparative analysis. This approach entails estimating nuisance parameters in a manner that minimizes bias in the  DR  estimator when working models are misspecified \citep{vermeulen2015bias}.   The study of these issues is beyond the scope of this paper, and we leave them as potential future research topics.

 \bibliographystyle{apalike}
					\bibliography{mybib}

\newpage

\renewcommand{\theproposition}{S\arabic{proposition}}
\renewcommand{\thetheorem}{S\arabic{theorem}}
\renewcommand{\theassumption}{S\arabic{assumption}}
\renewcommand{\thesection}{S\arabic{section}}
\renewcommand{\theequation}{S\arabic{equation}}
\renewcommand{\thelemma}{S\arabic{lemma}} 
\renewcommand{\thetable}{S\arabic{table}} 
\renewcommand{\theproof}{S\arabic{proof}}

\setcounter{section}{0}
\setcounter{equation}{0}
{\centering \section*{Supplementray Material}}
	  In the Supplementary Material, we provide proofs of lemmas and theorems in
the main paper.  
   \section{The proof of Lemma \ref{lem: if}}
 \label{sec:proof-tps-if}
 \begin{proof} 
Let $\mu_z=\E(Y_z)$. Let $\tilde{{\beta}}$ 
denote an intermediate value between $\hat{\beta}$
and ${\beta}^\ast$,  and $Q_{{{\minone\beta}}}^{\dr}(X,Z ; {\beta} ) \triangleq \partial Q_{\minone}(X,Z ; {\beta} ) / \partial {\beta} $, therefore, 
\begin{align*} 
 \sqrt{n}&\{ {\tau}^{\dr}_{\minone}(\alpha_{\minzero},\hat{\beta} )-\tau\} \\
= & \sqrt{n}\left[\frac{1}{n} \sum_{i=1}^n\left\{\frac{Y_i Z_i}{e_{\minzero}({ X}_i)}-\frac{Z_i-e_{\minzero}({ X}_i)}{e_{\minzero}({ X}_i)} Q^\dr_{\minone}({ X}_i,1; \hat{\beta}  )\right\}-{  \mu_1}\right] \\
& ~~~-\sqrt{n}\left[\frac{1}{n} \sum_{i=1}^n\left\{\frac{Y_i (1-Z_i)}{1-e_{\minzero}({ X}_i)}-\frac{e_{\minzero}({ X}_i)-Z_i}{1-e_{\minzero}({ X}_i)} Q^\dr_{\minone}({ X}_i,0; \hat{\beta}  )\right\}-{  \mu_0}\right] \\
= & \frac{1}{\sqrt{n}} \sum_{i=1}^n\left\{\frac{Y_i Z_i}{e_{\minzero}({ X}_i)}-\frac{Z_i-e_{\minzero}({ X}_i)}{e_{\minzero}({ X}_i)} Q^\dr_{\minone}({ X}_i,1; {\beta} ^*)-{  \mu_1} -\frac{Z_i-e_{\minzero}({ X}_i)}{e_{\minzero}({ X}_i)} Q_{{{\minone\beta}}}^{\dr}({ X}_i,1; \tilde{{{\beta}}}) (\hat{{{\beta}}}-{\beta} ^*)\right\}\\
  & ~~~-\frac{1}{\sqrt{n}} \sum_{i=1}^n\left\{\frac{Y_i (1-Z_i)}{1-e_{\minzero}({ X}_i)}-\frac{e_{\minzero}({ X}_i)-Z_i}{1-e_{\minzero}({ X}_i)} Q^\dr_{\minone}({ X}_i,0; {\beta} ^*)-{  \mu_0} -\frac{e_{\minzero}({ X}_i)-Z_i}{1-e_{\minzero}({ X}_i)} Q_{{{\minone\beta}}}^{\dr}({ X}_i,0; \tilde{{{\beta}}}) (\hat{{{\beta}}}-{\beta} ^*)\right\} \\
= & \frac{1}{\sqrt{n}} \sum_{i=1}^n\left\{\frac{Y_i Z_i}{e_{\minzero}({ X}_i)}-\frac{Z_i-e_{\minzero}({ X}_i)}{e_{\minzero}({ X}_i)} Q^\dr_{\minone}({ X}_i,1; {\beta} ^*)-{  \mu_1}\right\}\\&~~~  - \frac{1}{\sqrt{n}} \sum_{i=1}^n\left\{\frac{Y_i (1-Z_i)}{1-e_{\minzero}({ X}_i)}-\frac{e_{\minzero}({ X}_i)-Z_i}{1-e_{\minzero}({ X}_i)} Q^\dr_{\minone}({ X}_i,0; {\beta} ^*)-{  \mu_0}\right\}+o_p(1) .
\end{align*}
The third equality in the above comes from the fact that
\begin{align*}
\begin{aligned}
 \frac{1}{\sqrt{n}} \sum_{i=1}^n \frac{Z_i-e_{\minzero}({ X}_i)}{e_{\minzero}({ X}_i)} Q_{{{\minone\beta}}}^{\dr}({ X}_i,1; \tilde{{{\beta}}}) (\hat{{{\beta}}}-{\beta} ^*) &=\sqrt{n}(\hat{{{\beta}}}-{\beta} ^*)   \frac{1}{n} \sum_{i=1}^n \frac{Z_i-e_{\minzero}({ X}_i)}{e_{\minzero}({ X}_i)} Q_{{{\minone\beta}}}^{\dr}({ X}_i,1; \tilde{{{\beta}}}),\\ \frac{1}{\sqrt{n}} \sum_{i=1}^n \frac{e_{\minzero}({ X}_i)-Z_i}{1-e_{\minzero}({ X}_i)} Q_{{{\minone\beta}}}^{\dr}({ X}_i,0; \tilde{{{\beta}}}) (\hat{{{\beta}}}-{\beta} ^*) &=\sqrt{n}(\hat{{{\beta}}}-{\beta} ^*)   \frac{1}{n} \sum_{i=1}^n \frac{e_{\minzero}({ X}_i)-Z_i}{1-e_{\minzero}({ X}_i)} Q_{{{\minone\beta}}}^{\dr}({ X}_i,0; \tilde{{{\beta}}}).
\end{aligned}
\end{align*}
Since $\sqrt{n}(\hat{{{\beta}}}-{\beta} ^*)=O_p(1)$, and note that
\begin{align*}
 \frac{1}{n} \sum_{i=1}^n  \frac{Z_i-e_{\minzero}({ X}_i)}{e_{\minzero}({ X}_i)} Q_{{{\minone\beta}}}^{\dr}({ X}_i,1; \tilde{{{\beta}}})  
\stackrel{p}{\rightarrow} & \E\left\{\frac{Z-e_{\minzero}({ X})}{e_{\minzero}({ X})} Q_{{{\minone\beta}}}^{\dr}({ X},1 ;  {\beta} ^*)\right\} \\
= & \E\left[\E\left\{\frac{Z-e_{\minzero}({ X})}{e_{\minzero}({ X})} Q_{{{\minone\beta}}}^{\dr}({ X},1 ;  {\beta} ^*) \bigg\vert X\right\}\right]=0,\\ \frac{1}{n} \sum_{i=1}^n  \frac{Z_i-e_{\minzero}({ X}_i)}{1-e_{\minzero}({ X}_i)} Q_{{{\minone\beta}}}^{\dr}({ X}_i,0; \tilde{{{\beta}}})  
\stackrel{p}{\rightarrow} & \E\left\{\frac{Z-e_{\minzero}({ X})}{1-e_{\minzero}({ X})} Q_{{{\minone\beta}}}^{\dr}({ X},0 ;  {\beta} ^*)\right\} \\
= & \E\left[\E\left\{\frac{Z-e_{\minzero}({ X})}{1-e_{\minzero}({ X})} Q_{{{\minone\beta}}}^{\dr}({ X},0 ;  {\beta} ^*) \bigg\vert X\right\}\right]=0.
\end{align*} 
\end{proof}   
  \section{The  asymptotic variance when   propensity score is fully known}
  \label{sec: asyvar-tps2} 
    \begin{proof} 
      Based on Lemma \ref{lem: if} and the law of total variance, its asymptotic variance is 
\begin{equation}
    \label{eq:var-form-2} 
\begin{aligned}
 &\Sigma ^{\drt}_{\minone}( {\beta}  )\\&=\mathrm{var}\left\{\frac{ZY}{e_{\minzero}({ X})}-\frac{(1- Z)Y}{ 1-e_{\minzero}({ X})}-   \frac{Z-e_{\minzero}({ X})}{e_{\minzero}({ X})}   Q_{\minone}^\dr({ X},1; {{\beta}})+\frac{e_{\minzero}({ X})-Z}{ 1-e_{\minzero}({ X})}   Q_{\minone}^\dr({ X},0; {{\beta}})\right\}\\&= \E\left(\operatorname { v a r } \left\{\frac{Y Z}{e_{\minzero}({ X})}-\frac{Z-e_{\minzero}({ X})}{e_{\minzero}({ X})}  Q_{\minone}^{\dr}({ X},1; {{\beta}}) -\frac{Y(1- Z)}{1-e_{\minzero}({ X})}+\frac{e_{\minzero}({ X})-Z}{1-e_{\minzero}({ X})}   Q_{\minone}^{\dr}({ X},0; {{\beta}})\bigg\vert { X}\right\}\right) \\
 & ~~~+\operatorname{var}\left(\E \left\{\frac{Y Z}{e_{\minzero}({ X})}-\frac{Z-e_{\minzero}({ X})}{e_{\minzero}({ X})}  Q_{\minone}^{\dr}({ X},1; {{\beta}}) -\frac{Y(1- Z)}{1-e_{\minzero}({ X})}+\frac{e_{\minzero}({ X})-Z}{1-e_{\minzero}({ X})}   Q_{\minone}^{\dr}({ X},0; {{\beta}}) \bigg\vert { X}\right\}\right) \\
= & (SI)+(SI I) . 
\end{aligned}  
\end{equation}
We first calculate the second term (SII) in \eqref{eq:var-form-2}. Notice that when the propensity score is correctly specified, we have\begin{align*}
    \E &\left\{\frac{Y Z}{e_{\minzero}({ X})}-\frac{Z-e_{\minzero}({ X})}{e_{\minzero}({ X})}  Q_{\minone}^{\dr}({ X},1; {{\beta}}) -\frac{Y(1- Z)}{1-e_{\minzero}({ X})}+\frac{e_{\minzero}({ X})-Z}{1-e_{\minzero}({ X})}   Q_{\minone}^{\dr}({ X},0; {{\beta}}) \bigg\vert { X}\right\}\\ &=    \E\left\{\frac{Y Z}{e_{\minzero}({ X})} -\frac{Y(1- Z)}{1-e_{\minzero}({ X})}\bigg\vert { X}\right\}\\&=\tau(X).
\end{align*}Therefore, the second term (SII) in \eqref{eq:var-form-2} equals $ \operatorname{var}\left\{\tau({ X})\right\}=\E\{\tau^2({ X})\}-\tau^2,$  which does not depend on ${{\beta}}$.   
We next calculate the first term (SI) in \eqref{eq:var-form-2}. It can be shown that  
    \begin{equation*}
\begin{aligned}
  \operatorname{var}&\left\{\frac{Y Z}{e_{\minzero}({ X})}-\frac{Z-e_{\minzero}({ X})}{e_{\minzero}({ X})}  Q_{\minone}^{\dr}({ X},1; {{\beta}}) -\frac{Y(1- Z)}{1-e_{\minzero}({ X})}+\frac{e_{\minzero}({ X})-Z}{1-e_{\minzero}({ X})}   Q_{\minone}^{\dr}({ X},0; {{\beta}}) \bigg\vert { X}\right\}\\
= & \E\begin{bmatrix}
    \begin{pmatrix}
    \left\{\dfrac{Z-e_{\minzero}({ X})}{e_{\minzero}({ X})}\right\}^2  Q_{\minone}^{2}({ X},1;{\beta})  + \left\{\dfrac{Z-e_{\minzero}({ X})}{1-e_{\minzero}({ X})}\right\}^2  Q_{\minone}^{2}({ X},0;{\beta}) \\+2Q_{\minone}^{\dr}({ X},1;{\beta})Q_{\minone}^{\dr}({ X},0;{\beta})
\end{pmatrix}\Bigg\vert { X}
\end{bmatrix}\\
& -2 \E\left[\left\{\frac{Z-e_{\minzero}({ X})}{e_{\minzero}({ X})}Q_{\minone}^{\dr}({ X},1;{\beta}) +\frac{Z-e_{\minzero}({ X})}{1-e_{\minzero}({ X})}   Q_{\minone}^{\dr}({ X},0; {{\beta}})\right\}\left\{\frac{Y Z}{e_{\minzero}({ X})}-\frac{Y(1- Z)}{1-e_{\minzero}({ X})}-\tau({ X}) \right\}  \bigg\vert { X}\right] \\
& +\E\left[\left\{\frac{Y Z}{e_{\minzero}({ X})}-\frac{Y(1- Z)}{1-e_{\minzero}({ X})}-\tau({ X})\right\}^2  \bigg\vert { X}\right] .
\end{aligned}
\end{equation*} We know that
    \begin{align*}
& \E\left\{\frac{Z-2 e_{\minzero}({ X}) Z+e_{\minzero}^2({ X})}{e_{\minzero}^2({ X})}   Q_{\minone}^{2}({ X},1;{\beta})   \bigg\vert { X}\right\} =\frac{1-e_{\minzero}({ X})}{e_{\minzero}({ X})}  Q_{\minone}^{2}({ X},1;{\beta}), \\
& \E\left\{\frac{Z-2 e_{\minzero}({ X}) Z+e_{\minzero}^2({ X})}{\{1-e_{\minzero}(X)\}^2}   Q_{\minone}^{2}({ X},0;{\beta})   \bigg\vert { X}\right\} =\frac{e_{\minzero}({ X})}{1-e_{\minzero}({ X})}  Q_{\minone}^{2}({ X},0;{\beta}) ,
\end{align*} 
    \begin{equation*}
\begin{aligned}
\E& \left[\left\{\frac{Z-e_{\minzero}({ X})}{e_{\minzero}({ X})}Q_{\minone}^{\dr}({ X},1;{\beta}) +\frac{Z-e_{\minzero}({ X})}{1-e_{\minzero}({ X})}   Q_{\minone}^{\dr}({ X},0; {{\beta}})\right\}\left\{\frac{Y Z}{e_{\minzero}({ X})}-\frac{Y(1- Z)}{1-e_{\minzero}({ X})}-\tau({ X})  \right\}  \bigg\vert { X}\right] \\
& =\E\left[\frac{Z-e_{\minzero}({ X})}{e_{\minzero}({ X})}\left\{\frac{Y Z}{e_{\minzero}({ X})}-\frac{Y(1- Z)}{1-e_{\minzero}({ X})}-\tau({ X}) \right\}Q_{\minone}^{\dr}({ X},1;{\beta})  \bigg\vert { X}\right] \\
&~~~ +\E\left[\frac{Z-e_{\minzero}({ X})}{1-e_{\minzero}({ X})}\left\{\frac{Y Z}{e_{\minzero}({ X})}-\frac{Y(1- Z)}{1-e_{\minzero}({ X})}-\tau({ X}) \right\}Q_{\minone}^{\dr}({ X},0;{\beta})  \bigg\vert { X}\right] \\
& = \E\left[\frac{Y Z\left\{1-e_{\minzero}({ X})\right\}}{e_{\minzero}^2({ X})} Q_{\minone}^{\dr}({ X},1;{\beta})  \bigg\vert { X}\right] -\E\left\{\frac{Z-e_{\minzero}({ X})}{e_{\minzero}({ X})} \tau({ X})  Q_{\minone}^{\dr}({ X},1;{\beta})  \bigg\vert { X}\right\} \\
&~~~ + \E\left[\frac{Y(1-Z)e_{\minzero}({ X}) }{\{1-e_{\minzero}({ X})\}^2} Q_{\minone}^{\dr}({ X},0;{\beta})  \bigg\vert { X}\right] -\E\left\{\frac{Z-e_{\minzero}({ X})}{1-e_{\minzero}({ X})} \tau({ X})  Q_{\minone}^{\dr}({ X},0;{\beta})  \bigg\vert { X} \right\} \\
&~~~ +\E\left\{\frac{Y (1-Z)}{ 1-e_{\minzero}({ X}) } Q_{\minone}^{\dr}({ X},1;{\beta})  \bigg\vert { X}\right\}+\E\left\{\frac{YZ}{  e_{\minzero}({ X}) } Q_{\minone}^{\dr}({ X},0;{\beta})  \bigg\vert { X}\right\}\\& = \frac{1-e_{\minzero}({ X})}{e_{\minzero}({ X})} Q _{\minzero}({ X},1)  Q_{\minone}^{\dr}({ X},1;{\beta}) +\frac{ e_{\minzero}({ X})}{1-e_{\minzero}({ X})}Q _{\minzero}({ X},0)  Q_{\minone}^{\dr}({ X},0;{\beta}) \\& ~~~+Q _{\minzero}({ X},0)  Q_{\minone}^{\dr}({ X},1;{\beta})  + Q _{\minzero}({ X},1)  Q_{\minone}^{\dr}({ X},0;{\beta}), 
\end{aligned}
\end{equation*}and\begin{align*}
    \E&\left[\left\{\frac{Y Z}{e_{\minzero}({ X})}-\frac{Y(1- Z)}{1-e_{\minzero}({ X})}-\tau({ X})\right\}^2 \bigg\vert   { X}  \right] =\frac{\E(Y_1^2\mid   { X}) }{e_{\minzero}({ X})} +\frac{\E(Y_0^2\mid   { X}) }{1-e_{\minzero}({ X})}-\tau^2({ X}).
\end{align*}
Therefore, the first term (SI) in \eqref{eq:var-form-2} equals
\begin{align*} 
 (SI)&=\operatorname{var}\left\{\frac{Y Z}{e_{\minzero}({ X})}-\frac{Z-e_{\minzero}({ X})}{e_{\minzero}({ X})}  Q_{\minone}^{\dr}({ X},1; {{\beta}}) -\frac{Y(1- Z)}{1-e_{\minzero}({ X})}+\frac{e_{\minzero}({ X})-Z}{1-e_{\minzero}({ X})}   Q_{\minone}^{\dr}({ X},0; {{\beta}}) \bigg\vert { X}\right\}\\
 &=  \E\left\{\frac{1-e_{\minzero}({ X})}{e_{\minzero}({ X})}  Q_{\minone}^{2}({ X},1;{\beta}) +\frac{e_{\minzero}({ X})}{1-e_{\minzero}({ X})}  Q_{\minone}^{2}({ X},0;{\beta}) +2Q_{\minone}^{\dr}({ X},1;{\beta})Q_{\minone}^{\dr}({ X},0;{\beta})\right\} \\
&~~~  -2\E\left\{ \frac{1-e_{\minzero}({ X})}{e_{\minzero}({ X})} Q _{\minzero}({ X},1) \cdot Q_{\minone}^{\dr}({ X},1;{\beta})\right\} -2\E\left\{Q _{\minzero}({ X},0) \cdot Q_{\minone}^{\dr}({ X},1;{\beta})\right\} \\&~~~-2\E\left\{\frac{ e_{\minzero}({ X})}{1-e_{\minzero}({ X})}Q _{\minzero}({ X},0) \cdot Q_{\minone}^{\dr}({ X},0;{\beta}) \right\}-2\E\left\{ Q _{\minzero}({ X},1) \cdot Q_{\minone}^{\dr}({ X},0;{\beta})\right\} \\
&~~~ +\E\left\{\frac{Y_1^2}{e_{\minzero}({ X})} \right\}+\E\left\{\frac{Y_0^2}{1-e_{\minzero}({ X})} \right\} -\E\{\tau^2({ X})\}. 
\end{align*} 
The asymptotic variance $\Sigma^{\drt}_\minone( {\beta})$ can be calculated as follows:
\begin{equation} 
\label{eq:var-form}  
\begin{aligned}
\Sigma^{\drt}_\minone( {\beta})&=(SI)+(SII)\\&= {\E\left\{\frac{Y_1^2}{e_{\minzero}({ X})} \right\}}+{\E\left(\frac{1-e_{\minzero}({ X})}{e_{\minzero}({ X})} \left[ \{Q_{\minone}^{\dr}({ X},1; {{\beta}})\}^2-2  Q_{\minone}^{\dr}({ X},1; {{\beta}}) Q _{\minzero}({ X},1)\right]\right)} \\&~~~+ {\E\left\{\frac{Y_0^2}{1-e_{\minzero}({ X})} \right\}}+{\E\left(\frac{e_{\minzero}({ X})}{1-e_{\minzero}({ X})} \left[ \{Q_{\minone}^{\dr}({ X},0; {{\beta}})\}^2-2  Q_{\minone}^{\dr}({ X},0; {{\beta}})Q _{\minzero}({ X},0)\right]\right)} \\
&~~~  -2\E \left\{Q _{\minzero}({ X},0) \cdot Q_{\minone}^{\dr}({ X},1;{\beta})\right\}  - 2\E\left\{ Q _{\minzero}({ X},1) \cdot Q_{\minone}^{\dr}({ X},0;{\beta})\right\}\\&~~~ + 2\E\left\{Q^{\dr}_{\minone} ({ X},1;{\beta}) \cdot Q_{\minone}^{\dr}({ X},0;{\beta})\right\}  -\tau^2.\end{aligned}  
\end{equation} 
%
%
Furthermore, through a straightforward calculation of the above expression, we know that the asymptotic variance $\Sigma^{\drt}_\minone( {\beta})$ can be equivalently expressed in the following two forms:
\begin{equation}
\label{eq:dr-var-2} \begin{aligned}
(I)&= \E\left[\frac{ \{Y_0-Q_{\minone}^{\dr}({ X},0;{\beta} )\}^2}{ 1-e_{\minzero}({ X})}\right]+\E\left[\frac{ \{Y_1-Q_{\minone}^{\dr}({ X},1;{\beta} )\}^2}{ e_{\minzero}({ X})}\right]-\tau^2 \\&~~~~ 
-\E\left[  \{ Q _{\minzero}({ X},1)-Q_{\minone}^{\dr}({ X},1; {{\beta}})\}-  \{Q _{\minzero}({ X},0)-Q_{\minone}^{\dr}({ X},0; {{\beta}})\} \right]^2 \\&~~~~ +\E  \left\{    Q _{\minzero}({ X},1)-Q _{\minzero}({ X},0 )  \right\}  ^2 ,\end{aligned} 
\end{equation}
and
\begin{equation}
\label{eq:dr-var-22} \begin{aligned}
({{II}})&= \E \left[  \sqrt{\frac{1-e_{\minzero}({ X})}{e_{\minzero}({ X})}}\{Q_{\minone}^{\dr}({ X},1; {{\beta}})- Q _{\minzero}({ X},1)\}+ \sqrt{\frac{e_{\minzero}({ X})}{1-e_{\minzero}({ X})}}\{Q_{\minone}^{\dr}({ X},0; {{\beta}})-Q _{\minzero}({ X},0)\} \right]^2  \\
&~~~ +\E\left\{\frac{Y_1^2}{e_{\minzero}({ X})} \right\}+\E\left\{\frac{Y_0^2}{1-e_{\minzero}({ X})} \right\} -\tau^2 
\\&~~~-\E \left\{  \sqrt{\frac{1-e_{\minzero}({ X})}{e_{\minzero}({ X})}} Q _{\minzero}({ X},1) + \sqrt{\frac{e_{\minzero}({ X})}{1-e_{\minzero}({ X})}} Q _{\minzero}({ X},0) \right\}^2 . \\&
\equiv \mathcal{M}_1(  {\beta})+ \mathcal{M}_2+ \mathcal{M}_3 ,   
\end{aligned} 
\end{equation}
where 
 \begin{align*} 
& \mathcal{M}_1( {{\beta}})= \E \begin{bmatrix}
    \sqrt{\dfrac{1-e_{\minzero}({ X})}{e_{\minzero}({ X})}}\{Q_{\minone}^{\dr}({ X},1; {{\beta}})- Q _{\minzero}({ X},1)\}+ \sqrt{\dfrac{e_{\minzero}({ X})}{1-e_{\minzero}({ X})}}\{Q_{\minone}^{\dr}({ X},0; {{\beta}})-Q _{\minzero}({ X},0)\}
\end{bmatrix}^2  , 
\\&
\mathcal{M}_2 =\E\left\{\frac{Y_1^2}{e_{\minzero}({ X})} \right\}+\E\left\{\frac{Y_0^2}{1-e_{\minzero}({ X})} \right\}-\tau^2,\\
&\mathcal{M}_3 =-\E \left\{  \sqrt{\frac{1-e_{\minzero}({ X})}{e_{\minzero}({ X})}} Q _{\minzero}({ X},1) +\sqrt{\frac{e_{\minzero}({ X})}{1-e_{\minzero}({ X})}} Q _{\minzero}({ X},0) \right\}^2.
\end{align*} 
  \end{proof} 
    \section{The proof of double robustness property}
 \label{lem:dr-drest} 
     When the propensity score is fully known, but the outcome model may not be, the left-hand side of the above estimating equation  converges in probability to the left-hand side of \eqref{eq: dr-opt-est}, hence, $\hat{{{\beta}}}^{\drt} \stackrel{p}{\rightarrow} {{\beta}}^{\drt}$. On the other hand, even when the outcome model is correctly specified but the propensity score may not be, the left-hand side of the above estimating equation converges in probability to,
    \begin{gather*}  
\E \begin{pmatrix}
\begin{bmatrix}
     \dfrac{e_{\minzero}({ X})}{e_{\minone}( {{ X}};{\alpha}^{\ast})} \sqrt{\dfrac{1-e_{\minone}( {{ X}};{\alpha}^{\ast})}{e_{\minone}( {{ X}};{\alpha}^{\ast})}}\left\{Q_{\minzero}(X;1)-Q_{\minone}^{\dr}({ X}, 1 ; {\beta})\right\} \\\addlinespace[1.5mm]+\dfrac{1-e_{\minzero}({ X})}{1-e_{\minone}( {{ X}};{\alpha}^{\ast})}\sqrt{\dfrac{e_{\minone}( {{ X}};{\alpha}^{\ast})}{1-e_{\minone}( {{ X}};{\alpha}^{\ast})}}\left\{Q_{\minzero}(X;0)-Q_{\minone}^{\dr}({ X}, 0 ; {\beta})\right\}
\end{bmatrix} \\\addlinespace[2mm]
\times\left\{\sqrt{\dfrac{1-e_{\minone}( {{ X}};{\alpha}^{\ast})}{e_{\minone}( {{ X}};{\alpha}^{\ast})}} Q_{\minone\beta}^{\dr}({ X}, 1 ; {\beta})+\sqrt{\dfrac{e_{\minone}( {{ X}};{\alpha}^{\ast})}{1-e_{\minone}( {{ X}};{\alpha}^{\ast})}} Q_{\minone\beta}^{\dr}({ X}, 0 ; {\beta})\right\}
 \end{pmatrix}    ,  
    \end{gather*}
{which equals 0 when ${{\beta}}={{\beta}}_0$,} thus $\hat{{{\beta}}}^{\drt} \stackrel{p}{\rightarrow} {{\beta}}_0$. The following lemma formally establishes the theoretical properties of the proposed estimator ${\tau^{\dr}_{\minone}(\alpha^*, \hat{\beta}^{\drt})}$. 
 
\begin{lemma}
Under Assumptions \ref{assump: positivity} and \ref{assump: negative-control},  
the proposed estimator  has the following properties: 
\begin{enumerate}[label=(\roman*),leftmargin=30pt]
\item  When either the propensity score is fully known or the outcome  model is correctly specified, ${ {\tau} ^{\dr}_{\minone}({\alpha}^\ast,\hat{\beta} ^{\drt})}$ is a consistent estimator.
\item When the propensity score   is fully known, ${ {\tau} ^{\dr}_{\minone}(  \alpha^*, \hat{\beta} ^{\drt})}$ achieves the smallest asymptotic variance among all estimators of the form in \eqref{eq: dr-X}.
\item When the propensity score is fully known and the outcome  model is correctly specified, ${ {\tau} ^{\dr}_{\minone}({ {\alpha}}^*,\hat{\beta} ^{\drt})}$ achieves the semiparametric efficiency bound.
\end{enumerate}
\end{lemma}

  \begin{proof} 
(i)  When the propensity score is fully known, i.e., $e_\minone(X ;{\alpha})=e_{\minzero}( { X})$, we have shown that $\hat{{\beta}}^{\drt} \stackrel{p}{\rightarrow} {\beta}^{\drt}$. It is straightforward to verify that
\begin{equation*}\begin{aligned}
 {\tau}_\minone&^{\dr}(\alpha_{\minzero},\hat{{\beta}}^\drt )\\& =  \mathbb{P}_n
 \begin{bmatrix}
   \dfrac{ZY-\{{Z-e_{\minzero}({ X})}\}Q_{{{\minone}}}^{\dr}({ X},1;\hat{{\beta}}^\drt)}{e_{\minzero}({ X})}  - \dfrac{(1- Z)Y-\{{e_{\minzero}({ X})}-Z\}Q_{{{\minone}}}^{\dr}({ X},0;\hat{{\beta}}^\drt)}{ 1 - e_{\minzero}({ X})}  
 \end{bmatrix}  \\
& =\mathbb{P}_n
 \begin{bmatrix}
   \dfrac{ZY-\{{Z-e_{\minzero}({ X})}\}Q_{{{\minone}}}^{\dr}({ X},1;{{\beta}}^\drt)}{e_{\minzero}({ X})}  - \dfrac{(1- Z)Y-\{{e_{\minzero}({ X})}-Z\}Q_{{{\minone}}}^{\dr}({ X},0;{{\beta}}^\drt)}{ 1 - e_{\minzero}({ X})}  
 \end{bmatrix} +o_p(1)\\&=  \E \begin{Bmatrix}
 \dfrac{ZY}{e_{\minzero}({ X})}  - \dfrac{(1- Z)Y }{ 1 - e_{\minzero}({ X})} 
 \end{Bmatrix}  -\E\begin{bmatrix}    \dfrac{Z-e_{\minzero}({ X})}{e_{\minzero}({ X})}   Q_{{{\minone}}}^{\dr}({ X},1; {{\beta}}^\drt)-\dfrac{e_{\minzero}({ X})-Z}{ 1-e_{\minzero}({ X})}   Q_{{{\minone}}}^{\dr}({ X},0; {{\beta}}^\drt)  \end{bmatrix} +o_p(1)\\&= \tau+o_p(1).
\end{aligned} \end{equation*} 
When the outcome model is correct, we have  $\hat{{\beta}}^{\drt} \stackrel{p}{\rightarrow} {{\beta}}^{\drt}={\beta}_0$ and $Q_{\minone}(   X,Z ; {\beta}^\drt)=Q_{\minzero}(   X,Z)$ for the true value  ${\beta}_0$, but the propensity score may not be, $e_\minone( { X};\alpha^\ast)\neq e_{\minzero}( { X})$.  
  Thus,\begin{align*}  \begin{aligned}
  {\tau}&^{\dr}_\minone(\alpha^\ast,\hat{{\beta}}^\drt )\\& =  \mathbb{P}_n
 \begin{bmatrix}
   \dfrac{ZY-\{{Z-e_{\minone}({ X};\alpha^\ast)}\}Q_{{{\minone}}}^{\dr}({ X},1;\hat{{\beta}}^\drt)}{e_\minone( { X};{\alpha^\ast})}  - \dfrac{(1- Z)Y-\{{e_\minone( { X};{\alpha^\ast})}-Z\}Q_{{{\minone}}}^{\dr}({ X},0;\hat{{\beta}}^\drt)}{ 1 - e_\minone( { X};{\alpha^\ast})}  
 \end{bmatrix}\\ 
& =\mathbb{P}_n\{Q_{\minzero}( { X},1 )-Q_{\minzero}( { X},0)\} +\mathbb{P}_n \begin{bmatrix}
   \dfrac{Z\{Y-Q_{\minzero}( { X},1) \}}{e_\minone( { X};{\alpha^\ast})}  -   \dfrac{(1-Z)\{Y-Q_{\minzero}( { X},0 )\}}{1-e_\minone( { X};{\alpha^\ast})} 
 \end{bmatrix} +o_p(1) \\&=\E\{Q_{\minzero}( { X},1 )-Q_{\minzero}( { X},0)\} +\E \begin{bmatrix}
   \dfrac{Z\{Y-Q_{\minzero}( { X},1) \}}{e_\minone( { X};{\alpha^\ast})}  -   \dfrac{(1-Z)\{Y-Q_{\minzero}( { X},0 )\}}{1-e_\minone( { X};{\alpha^\ast})} 
 \end{bmatrix} +o_p(1)\\&= \tau+o_p(1).
\end{aligned} \end{align*}  

(ii) Since $\hat{{\beta}}^{\drt} \stackrel{p}{\rightarrow} {\beta}^{\drt}$ when the propensity score is correct, by definition of ${\beta}^{\drt}$, it is trivial that $ {\tau}^{\dr}_\minone(\alpha_{\minzero},\hat{{\beta}}^\drt )$ achieves the smallest variance among class of estimators \eqref{eq: dr-X}.

(iii) 
When both models are correctly specified, the asymptotic variance of the estimator can be expressed as follows, utilizing the first term in   \eqref{eq:dr-var-2}: 
\begin{equation}
    \label{eq:seb}
\begin{aligned}
   \Sigma^{\drt}_\minone( {\beta}_\minzero)&= \E\left[\frac{ \{Y_0-Q_{\minone}^{\dr}({ X},0  )\}^2}{ 1-e_{\minzero}({ X})}\right]+\E\left[\frac{ \{Y_1-Q_{\minone}^{\dr}({ X},1 )\}^2}{ e_{\minzero}({ X})}\right] \\&~~~+\E  \left\{    Q _{\minzero}({ X},1)-Q _{\minzero}({ X},0 )  \right\}  ^2 -\tau^2   ,
\end{aligned}
\end{equation}
which is equivalent to the semiparametric efficiency bound \citep{Hanh1998Econometrica,hirano2003efficient}.
\end{proof}
  \section{The proof of  Theorem \ref{thm:tps}}
  \label{sec:avar-tps}
 \subsection{Preliminaries}
 \label{ssec:Preliminary}
 We define $O = (Y, Z, X)$, and let ${\theta}$ represent the set of unknown parameters involved in deriving estimators for $\tau$. The asymptotic variance of ${\tau}^{\dr}_\minone(\hat{\alpha},\hat{{\beta}})$ is determined by solving a system of M-estimating equations $\sum_{i=1}^n {m}({ O}_i, {\theta}) = 0$ \citep{stefanski2002calculus}. In this system, the last component of ${m}({ O}_i, {\theta})$ corresponds to the estimating equation for $\tau$. We denote $\theta_{\ast}$ as the value that satisfies $\mathbb{E}\left\{{m}({ O}_i,\theta_{\ast})\right\} = 0$. 
 According to standard M-estimation theory: $$\sqrt{n}(\hat{{\theta}}-\theta_{\ast}) \stackrel{D}{\rightarrow} N(0, {B}(\theta_{\ast})^{-1} {D}(\theta_{\ast})\left\{{B}(\theta_{\ast})^{-1}\right\}^{\T}),$$where ${B}({\theta}) = \mathbb{E}\left\{\partial {m}({O}_i, {\theta}) / \partial {\theta}^{\T}\right\}$ and ${D}({\theta}) = \mathbb{E}\left\{{m}({O}_i, {\theta}) {m}^{\T}({O}_i, {\theta})\right\}$. Consequently, the asymptotic variance of ${\tau}^{\dr}_\minone(\hat{\alpha},\hat{{\beta}})$  is located in the last, rightmost diagonal entry of the corresponding matrix   ${B}(\theta_{\ast})^{-1} {D}(\theta_{\ast}) \left\{{B}(\theta_{\ast})^{-1}\right\}^{\T}$. Without loss of generality, we assume that 
  \begin{align*}
{B}({{\theta}})=\begin{Bmatrix}
    {B}_1({{\theta}}) & {0} \\
{B}_2({{\theta}}) & -1
\end{Bmatrix}, \quad {D}({{\theta}})=\begin{bmatrix}
{D}_{11}({{\theta}}) & {D}_{12}({{\theta}}) \\
 {D}_{12}^\T ({{\theta}}) & {D}_{22}({{\theta}})
\end{bmatrix}
\end{align*}
By some simple algebras, we know that
 $$
        \sqrt{n}\{{ {\tau}^{\dr}_{\minone}( {\alpha} ,{\hat{\beta}}  )}-\tau\} \stackrel{d}{\rightarrow} N(0, \Lambda ^\ast),  $$
        where  ${{\Lambda}} ^\ast = {B}_2({{\theta^\ast}}) {B}_1^{-1}({{\theta^\ast}})  D_{11}({\theta^\ast}) \left\{{B}_1^{-1}(\theta_{\ast})\right\}^{\T}  {B}_2^{\T}({{\theta^\ast}}) -2{B}_2({{\theta^\ast}}) \left\{{B}_1^{-1}(\theta_{\ast})\right\}^{\T}{D}_{12}({{\theta^\ast}}) +{D}_{22}({{\theta^\ast}}).  $
In the following discussions, we use ${0}_{a \times b}$ to denote a zero matrix with $a$ rows and $b$ columns. Sometimes we omit the dimension when there is no confusion.

 \subsection{The formal proof of Theorem \ref{thm:tps}}
 We first rewrite Theorem \ref{thm:tps} in the main text as follows:
  \begin{theorem} 
    \label{thm:tps-dr-sm} Let ${\theta} ^{\drt}_\ast=\{({{\beta}}^\drt)^\T, \tau\}^\T$, where ${{\beta}}^{\drt}$ solves the following equation:
\begin{align*} 
\E  \begin{pmatrix}
\begin{bmatrix}
     \dfrac{Z}{e_{\minone}( {{ X}};{\alpha}^{\ast})} \sqrt{\dfrac{1-e_{\minone}( {{ X}};{\alpha}^{\ast})}{e_{\minone}( {{ X}};{\alpha}^{\ast})}}\left\{Y-Q_{\minone}^{\dr}({ X}, 1 ; {\beta}^{\drt})\right\}\\\addlinespace[1.5mm]+\dfrac{1-Z}{1-e_{\minone}( {{ X}};{\alpha}^{\ast})}\sqrt{\dfrac{e_{\minone}( {{ X}};{\alpha}^{\ast})}{1-e_{\minone}( {{ X}};{\alpha}^{\ast})}}\left\{Y-Q_{\minone}^{\dr}({ X}, 0 ; {\beta}^{\drt})\right\}
\end{bmatrix} \\\addlinespace[2mm]
\times\left\{\sqrt{\dfrac{1-e_{\minone}( {{ X}};{\alpha}^{\ast})}{e_{\minone}( {{ X}};{\alpha}^{\ast})}} Q_{\minone\beta}^{\dr}({ X}, 1 ; {\beta}^{\drt})+\sqrt{\dfrac{e_{\minone}( {{ X}};{\alpha}^{\ast})}{1-e_{\minone}( {{ X}};{\alpha}^{\ast})}} Q_{\minone\beta}^{\dr}({ X}, 0 ; {\beta}^{\drt})\right\}
 \end{pmatrix}   =0.    
\end{align*}  Given Assumptions \ref{assump: positivity} and \ref{assump: negative-control}, when either the propensity score is fully known or the outcome model is correctly specified,  we have,  $$
        \sqrt{n}\{{ {\tau}^{\dr}_{\minone}( {\alpha} ,{\hat{\beta}}^\drt  )}-\tau\} \stackrel{d}{\rightarrow} N(0, \Lambda_\minone^\drt),  $$
   where $ \Lambda_\minone^\drt={{\Lambda}^\drt}({\theta}_\ast^\drt )$,  
\begin{gather*}
  {B}_1^{\drt}({{\theta}}) =\E \left({\partial  {   m}_1^\drt}/{\partial\beta^\T}\right),~~
  {B}_2^{\drt}({{\theta}}) =\E\left({\partial  {   m}_2^\drt}/{\partial\beta^\T}\right),\\
{D}_{11}^\drt({{\theta}}) =\E\{  {   m}_1^\drt({    m}_1^\drt)^\T\},~~
{D}_{12}^\drt({{\theta}}) =\E(  {   m}_1^\drt  {  m}_2^\drt) , ~~
{D}_{22}^\drt({{\theta}}) =\E ( {  m}_2^\drt)^2,\\{{\Lambda}^\drt}({\theta} ) = {B}_2^{\drt}({{\theta}})\{{B}_1^{\drt}({{\theta}})\}^{-1}  D_{11}^{\drt}({\theta})\{{B}_1^{\drt}({{\theta}})\}^{-\T}\{{B}_2^{\drt}({{\theta}})\}^\T-2{B}_2^{\drt}({{\theta}}) \{{B}_1^{\drt}({{\theta}})\}^{-\T}{D}_{12}^{\drt}({{\theta}}) +{D}_{22}^{\drt}({{\theta}}).
\end{gather*}
  Here we suppress the notations by writing  ${m}_1 ^\drt\equiv    {   m}_1^\drt({ O};{\theta})$ and $ {m}_2 ^\drt\equiv    {   m}_2^\drt({ O};{\theta})$, where 
\begin{align*}
    m^\drt({ O};{\theta})&=\left\{         {   m}_1^\drt({ O};{\theta}),m_2^\drt({ O};{\theta})\right\}^\T,\\   {   m}_1^\drt({ O};{\theta}) & = \begin{pmatrix}
\begin{bmatrix}
     \dfrac{Z}{e_{\minone}( {{ X}};{\alpha}^{\ast})} \sqrt{\dfrac{1-e_{\minone}( {{ X}};{\alpha}^{\ast})}{e_{\minone}( {{ X}};{\alpha}^{\ast})}}\left\{Y-Q_{\minone}^{\dr}({ X}, 1 ; {\beta})\right\}\\\addlinespace[1.5mm]+\dfrac{1-Z}{1-e_{\minone}( {{ X}};{\alpha}^{\ast})}\sqrt{\dfrac{e_{\minone}( {{ X}};{\alpha}^{\ast})}{1-e_{\minone}( {{ X}};{\alpha}^{\ast})}}\left\{Y-Q_{\minone}^{\dr}({ X}, 0 ; {\beta})\right\}
\end{bmatrix} \\\addlinespace[2mm]
\times\left\{\sqrt{\dfrac{1-e_{\minone}( {{ X}};{\alpha}^{\ast})}{e_{\minone}( {{ X}};{\alpha}^{\ast})}} Q_{\minone\beta}^{\dr}({ X}, 1 ; {\beta})+\sqrt{\dfrac{e_{\minone}( {{ X}};{\alpha}^{\ast})}{1-e_{\minone}( {{ X}};{\alpha}^{\ast})}} Q_{\minone\beta}^{\dr}({ X}, 0 ; {\beta})\right\}
 \end{pmatrix} ,\\
  {   m}_2^\drt({ O};{\theta})&=   \frac{ZY}{e_{\minone}( {{ X}};{\alpha}^{\ast})}-\frac{(1- Z)Y}{ 1-e_{\minone}( {{ X}};{\alpha}^{\ast})}-   \frac{Z-e_{\minone}( {{ X}};{\alpha}^{\ast})}{e_{\minone}( {{ X}};{\alpha}^{\ast})}   Q_{\minone}^\dr({ X},1; {{\beta}})\\&~~~+\frac{e_{\minone}( {{ X}};{\alpha}^{\ast})-Z}{ 1-e_{\minone}( {{ X}};{\alpha}^{\ast})}   Q_{\minone}^\dr({ X},0; {{\beta}})-\tau . ~~~~~~~~~~~~~~~~~~~~~~~~~~~
\end{align*}

\end{theorem}
  We next present the proof of  Theorem \ref{thm:tps} as follows:
\begin{proof} 
 
  For the  improved DR estimator ${ {\tau} _\minone(  {\alpha} ,\hat{\beta}^\drt  )}$,  by M-estimation theory,  the proposed estimators ${\hat{\theta}}_\ast^{\drt}=\left\{({\hat{\beta}}^\drt)^\T, \tau^\dr_\minone(  {\alpha} ,\hat{\beta}^\drt  )\right\}^\T$ are  asymptotically normal with 
\begin{align*}
\sqrt{n}(\hat{{{\theta}}}^{\drt}-{{\theta}}_\ast^{\drt}) \stackrel{D}{\rightarrow} N\left(0,\left\{{B}^{\drt}({{\theta}}_\ast^{\drt})\right\}^{-1} {D}^{\drt}({{\theta}}_\ast^{\drt}) \left\{{B}^{\drt}({{\theta}}_\ast^{\drt})\right\}^{-\T} \right),
\end{align*}
where
\begin{align*}
{B}^{\drt}({{\theta}})=\begin{Bmatrix}
    {B}_1^{\drt}({{\theta}}) & {0}  \\
{B}_2^{\drt}({{\theta}}) & -1
\end{Bmatrix}, \quad {D}^{\drt}({{\theta}})=\begin{bmatrix}
{D}_{11}^{\drt}({{\theta}}) & {D}_{12}^{\drt}({{\theta}}) \\
\left\{{D}_{12}^{\drt}({{\theta}})\right\}^\T & {D}_{22}^{\drt}({{\theta}})
\end{bmatrix}
\end{align*}
with all the quantities defined in the theorem. The rest is by the results in Section \ref{ssec:Preliminary}.

\end{proof}
\section{The proof of Theorem  \ref{thm:com-drt}}
\label{sec:var-tps}
\begin{proof} 
\label{proof:var-tps-aipw}
We proceed the proof as follows:
\begin{enumerate}
    
\item  Under Assumption \ref{cond:two-res-aipw}(i),    we represent an extended vector as $\tilde{\delta} = \left\{({\hat\beta} ^\drt)^\T, 0_{1\times m} \right\}^\T$ to accommodate the dimension of $\delta$, where $m$ is the dimension of the parameters involving $W$ that need to be expanded. Therefore, $Q_{\minone}^\dr({ X},Z; {{\hat\beta}^\drt})=Q_{\mintwo}^\dr( X  ,W,Z;\tilde{\delta})$. Moreover, according to \eqref{eq: dr-X} and \eqref{eq: dr-XW}, we know that 
    \begin{equation*}
 \begin{aligned}
         {\tau}^{\dr}_{\minone}(  {\alpha}_{\minzero}, {\hat\beta} ^\drt ) &=  
 \mathbb{P}_n  \begin{bmatrix}
 \dfrac{ZY-\{{Z-e_{\minzero}(  {{ X}} )}\}Q^{\dr}_{\minone}( {{ X}},1;{\hat\beta} ^\drt)}{e_{\minzero}( {{ X}} )} \\\addlinespace[1mm] - \dfrac{(1- Z)Y-\{{e_{\minzero}( {{ X}},W)}-Z\}Q^{\dr}_{\minone}( {{ X}},0; {\hat\beta}^\drt)}{ 1 - e_{\minzero}( {{ X}} )}  
 \end{bmatrix},\\
 {\tau}^{\dr}_{\mintwo}(  {\alpha}_{\minzero},  {\gamma}_{\minzero}, \tilde{\delta}  )& =  
 \mathbb{P}_n  \begin{bmatrix}
 \dfrac{ZY-\{{Z-e_{\minzero}(  {{ X}},W )}\}Q_{\mintwo}^\dr( X  ,W,1;\tilde{\delta})}{e_{\minzero}( {{ X}},W )} \\\addlinespace[1mm] - \dfrac{(1- Z)Y-\{{e_{\minzero}( {{ X}},W)}-Z\}Q_{\mintwo}^\dr( X  ,W,0;\tilde{\delta})}{ 1 - e_{\minzero}( {{ X}} ,W)}  
 \end{bmatrix}.  
 \end{aligned}
\end{equation*} Since $e_{\minzero}({{X}}) = e_{\minzero}({{X}},W)$, it is obvious that ${\tau}^{\dr}_{\minone}({\alpha}_{\minzero}, {\hat\beta}^\drt)$ is exactly equal to ${\tau}^{\dr}_{\mintwo}({\alpha}_{\minzero}, {\gamma}_{\minzero}, \tilde{\delta})$. It follows that ${\tau}^{\dr}_{\minone}({\alpha}_{\minzero}, {\hat\beta}^\drt)$ and ${\tau}^{\dr}_{\mintwo}({\alpha}_{\minzero}, {\gamma}_{\minzero}, \tilde{\delta})$ should have exactly the same asymptotic variance, i.e., $\Sigma^\drt_{\minone}({\beta}^\drt) = \Sigma^\drt_{\mintwo}(\tilde{\delta}^\ast)$, where $\tilde{\delta}^\ast = ({\beta}^\drt, 0_{1\times m})^\T$. 

  Therefore, by the definition of ${\delta}^\drt$, we have that   $  \Sigma_\minone^\drt( {\beta}^\drt )=\Sigma_\mintwo^\drt ({\tilde\delta} ) \geq \Sigma_\mintwo^\drt  ( {\delta}^\drt )   $. 
  
  \item  Under Assumption \ref{cond:two-res-aipw}(i), we know that $Q_{\minone}(   X,W,Z ; {\delta}^\drt)=Q_{\minzero}(   X,W,Z)$. Using the second form in \eqref{eq:dr-var-22}, we   have\begin{align*} 
 \Sigma_\mintwo^\drt ( {\delta}^\drt)  =&\E\left\{\frac{Y_1^2}{e_{\minzero}({ X},W)} \right\}+\E\left\{\frac{Y_0^2}{1-e_{\minzero}({ X},W)} \right\} -\tau^2 
\\&~~~-\E\left[\left\{  \sqrt{\frac{1-e_{\minzero}({ X},W)}{e_{\minzero}({ X},W)}} Q _{\minzero}({ X},W,1) + \sqrt{\frac{e_{\minzero}({ X},W)}{1-e_{\minzero}({ X},W)}} Q _{\minzero}({ X},W,0) \right\}^2\right].
  \end{align*}Moreover, let $ \Sigma^{\mathrm{eff}} _\minone
 $ represent  the semiparametric efficiency bound corresponding to the estimators defined based only on \(X\) \citep{Hanh1998Econometrica,hirano2003efficient}, that is,\begin{align*}
     \Sigma^{\mathrm{eff}}_\minone
  =  &\E\left\{\frac{Y_1^2}{e_{\minzero}({ X})} \right\}+\E\left\{\frac{Y_0^2}{1-e_{\minzero}({ X})} \right\} -\tau^2 
\\&~~~-\E\left[\left\{  \sqrt{\frac{1-e_{\minzero}({ X})}{e_{\minzero}({ X})}} Q _{\minzero}({ X},1) + \sqrt{\frac{e_{\minzero}({ X})}{1-e_{\minzero}({ X})}} Q _{\minzero}({ X},0) \right\}^2\right].
 \end{align*}  
Since $e_{\minzero}(X)=e_{\minzero}(X,W)$, by Jensen's inequality, we know that 
\begin{gather*}
    \E\left[\left\{  \sqrt{\frac{1-e_{\minzero}({ X},W)}{e_{\minzero}({ X},W)}} Q _{\minzero}({ X},W,1) + \sqrt{\frac{e_{\minzero}({ X},W)}{1-e_{\minzero}({ X},W)}} Q _{\minzero}({ X},W,0) \right\}^2\right]\\
   \leq  \E\left[\left\{  \sqrt{\frac{1-e_{\minzero}({ X})}{e_{\minzero}({ X})}} Q _{\minzero}({ X},1) + \sqrt{\frac{e_{\minzero}({ X})}{1-e_{\minzero}({ X})}} Q _{\minzero}({ X},0) \right\}^2\right],
\end{gather*}
and hence,
\begin{align*}
  \Sigma^\drt_\minone ( {\beta}^\drt)\geq    \Sigma^{\mathrm{eff}}_\minone
  \geq  \Sigma_\mintwo^\drt  ( {\delta}^\drt )  .
\end{align*}
\end{enumerate}
\end{proof} 
\section{The proof of Lemma \ref{lem:if-eps}} 
\label{sec:der-est-expression}\begin{proof} 
  When the propensity score model is correctly specified, we have
  \begin{align*} 
 \sqrt{n} \{ {\tau}^{\dr}_{\minone}(\hat{\alpha},\hat{\beta} )-\tau\}  
&=  \sqrt{n}\left[\frac{1}{n} \sum_{i=1}^n\left\{\frac{Y_i Z_i}{e_\minone({ X}_i;\hat{\alpha})}-\frac{Z_i-e_\minone({ X}_i;\hat{\alpha}) }{e_\minone({ X}_i;\hat{\alpha})} Q_{\minone}({ X}_i,1; \hat{\beta}  )\right\}-{  \mu_1}\right] \\
& ~~~~-\sqrt{n}\left[\frac{1}{n} \sum_{i=1}^n\left\{\frac{Y_i (1-Z_i)}{1-e_\minone({ X}_i;\hat{\alpha})}-\frac{ e_\minone({ X}_i;\hat{\alpha})-Z_i}{1-e_\minone({ X}_i;\hat{\alpha})} Q_{\minone}({ X}_i,0; \hat{\beta}  )\right\}-{  \mu_0}\right] \\
 & =\frac{1}{\sqrt{n}} \sum_{i=1}^n\left\{\frac{Y_i Z_i}{e_\minone({ X}_i;\hat{\alpha})}-\frac{Z_i-e_\minone({ X}_i;\hat{\alpha}) }{e_\minone({ X}_i;\hat{\alpha})} Q_{\minone}({ X}_i,1; {\beta} ^*)-{  \mu_1}\right\}\\&~~~~  - \frac{1}{\sqrt{n}} \sum_{i=1}^n\left\{\frac{Y_i (1-Z_i)}{1-e_\minone({ X}_i;\hat{\alpha})}-\frac{e_\minone({ X}_i;\hat{\alpha})-Z_i}{1-e_\minone({ X}_i;\hat{\alpha})} Q^\dr_{\minone}({ X}_i,0; {\beta} ^*)-{  \mu_0}\right\}+o_p(1) .
\end{align*} 
Now we expand $\hat{{\alpha}}$ about ${\alpha}_0$ to obtain
\begin{align*}
\begin{aligned}
 \sqrt{n} \{ {\tau}^{\dr}_{\minone}(\hat{\alpha},\hat{\beta} )-\tau\} & =\frac{1}{\sqrt{n}} \sum_{i=1}^n {\varphi}^\dr_{\minone}(Y_i, Z_i, { X}_i;{\alpha}_0, {{\beta}}^*) \\
& ~~~+\left\{\frac{1}{n} \sum_{i=1}^n \frac{\partial {\varphi}^\dr_{\minone}(Y_i, Z_i, { X}_i;{\alpha}^*, {{\beta}}^*)}{\partial {\alpha}^\T}\right\} \sqrt{n}(\hat{{\alpha}}-{\alpha}_0)+o_p(1),
\end{aligned}
\end{align*}
where $\tilde{\alpha}$ is some intermediate value between $\hat{{\alpha}}$ and ${\alpha}_0$. Since under regularity conditions, $\tilde{\alpha}$ 
 converges in probability to ${\alpha}_0$, we obtain
\begin{align*}
\frac{1}{n} \sum_{i=1}^n \frac{\partial {\varphi}^\dr_{\minone}(Y_i, Z_i, { X}_i;\tilde{\alpha}, {{\beta}}^*)}{\partial {\alpha}^\T} \stackrel{p}{\rightarrow} \E\left\{\frac{\partial {\varphi}^\dr_{\minone}(Y_i, Z_i, { X}_i;{\alpha}_0, {{\beta}}^*)}{\partial {\alpha}^\T}\right\} .
\end{align*}
Using standard results from finite-dimensional parametric models, we know that
\begin{align*}
\sqrt{n}(\hat{{\alpha}}-{\alpha}_0)=\frac{1}{\sqrt{n}} \sum_{i=1}^n\left[\E\left\{S_\alpha(Z_i,{ X}_i;{\alpha}_0) S^\T_\alpha(Z_i,{ X}_i;{\alpha}_0)\right\}\right]^{-1} S_\alpha(Z_i,{ X}_i;{\alpha}_0)+o_p(1).
\end{align*}
We hence deduce that the influence function is
\begin{align*}
\begin{aligned}
 \sqrt{n}& \{ {\tau}^{\dr}_{\minone}(\hat{\alpha},\hat{\beta} )-\tau\} \\=& \frac{1}{\sqrt{n}} \sum_{i=1}^n {\varphi}^\dr_{\minone}(Y_i, Z_i, { X}_i;{\alpha}_0, {{\beta}}^*) \\
& ~~+\frac{1}{\sqrt{n}} \sum_{i=1}^n\left(\E\left\{\frac{\partial {\varphi}^\dr_{\minone}(Y_i, Z_i, { X}_i;{\alpha}_0, {{\beta}}^*)}{\partial {\alpha}^{\T}} \right\} \left[\E\left\{S_\alpha(Z_i,{ X}_i;{\alpha}_0) S^{\T}_\alpha(Z_i,{ X}_i;{\alpha}_0)\right\}\right]^{-1} S_\alpha(Z_i,{ X}_i;{\alpha}_0) \right)\\&~~+o_p(1) .
\end{aligned}
\end{align*}
Thus, the influence function is
\begin{align*}
\varphi^\dr_{\minone}(Y, Z, { X};{\alpha}_0, {{\beta}}^*)+\E\left\{\frac{\partial \varphi^\dr_{\minone}(Y, Z, { X};{\alpha}_0, {{\beta}}^*)}{\partial {\alpha}^{\T}} \right\} \left[\E\left\{S_\alpha(Z, { X};{\alpha}_0) S_\alpha^{\T}(Z, { X};{\alpha}_0)\right\}\right]^{-1} S_\alpha(Z, { X};{\alpha}_0) .
\end{align*}
\end{proof} 

  \section{The  asymptotic variance when   propensity score is correctly specified}
 \begin{proof}
 \label{proof:der-est-expression}
     The inﬂuence function  $ \tilde   \varphi^\dr_{\minone}(Y, Z, { X};{\alpha}_0, {{\beta}}^*)$  can be rewritten as
\begin{align*}
  \tilde   \varphi^\dr_{\minone}(Y, Z, { X};{\alpha}_0, {{\beta}}^*)&=\frac{ZY}{e_{\minzero}({ X})}-\frac{(1- Z)Y}{ 1-e_{\minzero}({ X})}-\tau   -   \frac{Z-e_{\minzero}({ X}) }{e_{\minzero}({ X})}   Q_{\minone}^\dr({ X},1; {{\beta}}^*)\\&~~~~+\frac{  e_{\minzero}({ X})-Z }{ 1-e_{\minzero}({ X})}   Q_{\minone}^\dr({ X},0; {{\beta}}^*) -{\mathcal{D}}_{\minone}^\dr( {{\beta}}^*) \mathcal{H}_{{\alpha} {\alpha},  \minzero}^{-1} S_{\alpha}(Z,{ X}; {\alpha}_0).
\end{align*}We hence have 
    \begin{align*}
  \Sigma_\minone^\dre (\beta) &=  \mathrm{var}\{   \tilde   \varphi^\dr_{\minone}(Y, Z, { X};{\alpha}_0, {\beta})\}\\&=\E\{   \tilde  \varphi^\dr_{\minone}(Y, Z, { X};{\alpha}_0, {\beta})\}^2\\&=\underbrace{\E\left\{\frac{ZY}{e_{\minzero}({ X})}-\frac{(1-Z)Y}{1-e_{\minzero}({ X})}\right\}^2}_{B_1}-\tau ^2\\&~~~+\underbrace{\E\left\{  \frac{Z-e_{\minzero}({ X}) }{e_{\minzero}({ X})}   Q_{\minone}^\dr({ X},1; {\beta})-\frac{ e_{\minzero}({ X})-Z }{ 1-e_{\minzero}({ X})}   Q_{\minone}^\dr({ X},0; {\beta})+{\mathcal{D}}_{\minone}^\dr( {\beta}) \mathcal{H}_{{\alpha} {\alpha},  \minzero}^{-1} S_{\alpha}(Z,{ X}; {\alpha}_0)\right\}^2}_{B_2(\beta) }\\&~~~-\underbrace{2\E\left[  \begin{Bmatrix}
     \dfrac{Z-e_{\minzero}({ X}) }{e_{\minzero}({ X})}   Q_{\minone}^\dr({ X},1; {\beta})-\dfrac{e_{\minzero}({ X})-Z }{ 1-e_{\minzero}({ X})}   Q_{\minone}^\dr({ X},0; {\beta})\\\addlinespace[1mm]+{\mathcal{D}}_{\minone}^\dr( {\beta}) \mathcal{H}_{{\alpha} {\alpha},  \minzero}^{-1} S_{\alpha}(Z,{ X}; {\alpha}_0) 
  \end{Bmatrix}\times\left\{\frac{ZY}{e_{\minzero}({ X})}-\frac{(1- Z)Y}{ 1-e_{\minzero}({ X})} \right\}\right]}_{B_3(\beta) }.
\end{align*}
After some calculations, it can be shown that     \begin{align*}
    &  \frac{Z-e_{\minzero}({ X}) }{e_{\minzero}({ X})}   Q_{\minone}^\dr({ X},1; {\beta})-\frac{ e_{\minzero}({ X})-Z }{ 1-e_{\minzero}({ X})}   Q_{\minone}^\dr({ X},0; {\beta})+{\mathcal{D}}_{\minone}^\dr( {\beta}) \mathcal{H}_{{\alpha} {\alpha},  \minzero}^{-1} S_{\alpha}(Z,{ X}; {\alpha}_0)\\&=\frac{Z-e_{\minzero}({ X}) }{e_{\minzero}({ X})}   Q_{\minone}^\dr({ X},1; {\beta})-\frac{ e_{\minzero}({ X})-Z }{ 1-e_{\minzero}({ X})}   Q_{\minone}^\dr({ X},0; {\beta}) \\&~~~~ +{\mathcal{D}}_{\minone}^\dr( {\beta}) \mathcal{H}_{{\alpha} {\alpha},  \minzero}^{-1}  \left\{\frac{Ze_{\minone\alpha}({ X};{\alpha}_{\minzero})}{e_{\minzero}({ X})}-\frac{(1-Z)e_{\minone\alpha}({ X};{\alpha}_{\minzero})}{1-e_{\minzero}({ X})}\right\}
\\&= \frac{Z }{e_{\minzero}({ X})}  \left\{ Q_{\minone}^\dr({ X},1; {\beta})+{\mathcal{D}}_{\minone}^{ \dr}( {\beta}) \mathcal{H}_{{\alpha} {\alpha},  \minzero}^{-1} e_{\minone\alpha}({ X};{\alpha}_{\minzero}) \right\} -Q_{\minone}^\dr({ X},1;
{\beta}) \\&~~~~-\frac{ 1-Z }{ 1-e_{\minzero}({ X})} \left\{   Q_{\minone}^\dr({ X},0; {\beta})+{\mathcal{D}}_{\minone}^{ \dr}( {\beta}) \mathcal{H}_{{\alpha} {\alpha},  \minzero}^{-1}  e_{\minone\alpha}({ X};{\alpha}_{\minzero}) \right\} +Q_{\minone}^\dr({ X},0;
{\beta})\\&=   \dfrac{Z-e_{\minzero}(X) }{e_{\minzero}({ X})}  \left\{ Q_{\minone}^\dr({ X},1; {\beta})+{\mathcal{D}}_{\minone}^\dr( {\beta}) \mathcal{H}_{{\alpha} {\alpha},  \minzero}^{-1} e_{\minone\alpha}({ X};{\alpha}_{\minzero}) \right\}\\&~~~~-\dfrac{ e_{\minzero}(X)-Z }{ 1-e_{\minzero}({ X})} \left\{   Q_{\minone}^\dr({ X},0; {\beta})+{\mathcal{D}}_{\minone}^\dr( {\beta}) \mathcal{H}_{{\alpha} {\alpha},  \minzero}^{-1}  e_{\minone\alpha}({ X};{\alpha}_{\minzero}) \right\} . \\
    B_2&(\beta)   =
    \E\begin{bmatrix}
         \dfrac{Z-e_{\minzero}(X) }{e_{\minzero}({ X})}  \left\{ Q_{\minone}^\dr({ X},1; {\beta})+{\mathcal{D}}_{\minone}^\dr( {\beta}) \mathcal{H}_{{\alpha} {\alpha},  \minzero}^{-1} e_{\minone\alpha}({ X};{\alpha}_{\minzero}) \right\}\\-\dfrac{ e_{\minzero}(X)-Z }{ 1-e_{\minzero}({ X})} \left\{   Q_{\minone}^\dr({ X},0; {\beta})+{\mathcal{D}}_{\minone}^\dr( {\beta}) \mathcal{H}_{{\alpha} {\alpha},  \minzero}^{-1}  e_{\minone\alpha}({ X};{\alpha}_{\minzero}) \right\}   
    \end{bmatrix}^2 .
\\
    B_3& (\beta) =2\E\begin{pmatrix}
          \begin{bmatrix}
        \dfrac{Z-e_{\minzero}({ X}) }{e_{\minzero}({ X})}  \left\{ Q_{\minone}^\dr({ X},1; {\beta})+{\mathcal{D}}_{\minone}^\dr( {\beta}) \mathcal{H}_{{\alpha} {\alpha},  \minzero}^{-1} e_{\minone\alpha}({ X};{\alpha}_{\minzero}) \right\}\\-\dfrac{ e_{\minzero}({ X})-Z }{ 1-e_{\minzero}({ X})} \left\{   Q_{\minone}^\dr({ X},0; {\beta})+{\mathcal{D}}_{\minone}^\dr( {\beta}) \mathcal{H}_{{\alpha} {\alpha},  \minzero}^{-1}  e_{\minone\alpha}({ X};{\alpha}_{\minzero}) \right\} 
    \end{bmatrix}   \times \left\{\dfrac{ZY}{e_{\minzero}({ X})}-\dfrac{(1- Z)Y}{ 1-e_{\minzero}({ X})} \right\}
    \end{pmatrix}   .
\end{align*}
By taking the derivative with respect to ${\beta}$ and set it equal to 0, we know that the variance is minimized by the solution to the following estimating  equation,
\begin{align*}
    \dfrac{\partial   \Sigma_\minone^\dre (\beta) }{\partial\beta}=   \dfrac{\partial   B_2 (\beta) }{\partial\beta}-\dfrac{\partial   B_3 (\beta) }{\partial\beta}=0,
\end{align*}
which is equivalent to the following equation,
\begin{align*}
    \E& \begin{pmatrix}
        &  \begin{bmatrix}
      \dfrac{Z Y }{e_{\minzero}({ X})}    -    \dfrac{Z-e_{\minzero}({ X}) }{e_{\minzero}({ X})}  \left\{ Q_{\minone}^\dr({ X},1; {\beta})+{\mathcal{D}}_{\minone}^\dr( {\beta}) \mathcal{H}_{{\alpha} {\alpha},  \minzero}^{-1} e_{\minone\alpha}({ X};{\alpha}_{\minzero}) \right\}\\\addlinespace[1mm]-  \dfrac{(1-Z) Y }{1-e_{\minzero}({ X})}    +\dfrac{ e_{\minzero}({ X})-Z }{ 1-e_{\minzero}({ X})} \left\{  Q_{\minone}^\dr({ X},0; {\beta})+{\mathcal{D}}_{\minone}^\dr( {\beta}) \mathcal{H}_{{\alpha} {\alpha},  \minzero}^{-1}  e_{\minone\alpha}({ X};{\alpha}_{\minzero}) \right\} 
    \end{bmatrix}  \\\addlinespace[1mm]&\times \begin{bmatrix}
        \dfrac{Z-e_{\minzero}({ X}) }{e_{\minzero}({ X})}  \left\{ Q^\dr_{{\minone\beta}}({ X},1; {{\beta}})+{\mathcal{D}}_{\minone\beta}^\dr( {{\beta}}) {\mathcal{H}}_{{\alpha} {\alpha},  \minzero}^{-1} e_{\minone\alpha}({ X};{\alpha}_{\minzero}) \right\}\\\addlinespace[1mm]-\dfrac{ e_{\minzero}({ X})-Z }{ 1-e_{\minzero}({ X})}   \left\{ Q^\dr_{{\minone\beta}}({ X},0; {{\beta}})+{\mathcal{D}}_{\minone\beta}^\dr( {{\beta}}) {\mathcal{H}}_{{\alpha} {\alpha},  \minzero}^{-1} e_{\minone\alpha}({ X};{\alpha}_{\minzero}) \right\}
    \end{bmatrix} 
    \end{pmatrix} \\ \addlinespace[1.5mm]  & =\E\begin{pmatrix}
        &  \begin{bmatrix}
    \dfrac{Z }{e_{\minzero}({ X})}  \left\{Y- Q_{\minone}^\dr({ X},1; {\beta})-{\mathcal{D}}_{\minone}^\dr( {\beta}) \mathcal{H}_{{\alpha} {\alpha},  \minzero}^{-1} e_{\minone\alpha}({ X};{\alpha}_{\minzero}) \right\}\\\addlinespace[1mm] - \dfrac{1-Z }{ 1-e_{\minzero}({ X})} \left\{  Y-Q_{\minone}^\dr({ X},0; {\beta})-{\mathcal{D}}_{\minone}^\dr( {\beta}) \mathcal{H}_{{\alpha} {\alpha},  \minzero}^{-1}  e_{\minone\alpha}({ X};{\alpha}_{\minzero}) \right\} \\\addlinespace[1mm] +\{Q_{\minone}^\dr({ X},1; {\beta})-Q_{\minone}^\dr({ X},0; {\beta})\}
    \end{bmatrix}  \\\addlinespace[1mm]&\times \begin{bmatrix}
        \dfrac{Z-e_{\minzero}({ X}) }{e_{\minzero}({ X})}  \left\{ Q^\dr_{{\minone\beta}}({ X},1; {{\beta}})+{\mathcal{D}}_{\minone\beta}^\dr( {{\beta}}) {\mathcal{H}}_{{\alpha} {\alpha},  \minzero}^{-1} e_{\minone\alpha}({ X};{\alpha}_{\minzero}) \right\}\\\addlinespace[1mm]-\dfrac{ e_{\minzero}({ X})-Z }{ 1-e_{\minzero}({ X})}   \left\{ Q^\dr_{{\minone\beta}}({ X},0; {{\beta}})+{\mathcal{D}}_{\minone\beta}^\dr( {{\beta}}) {\mathcal{H}}_{{\alpha} {\alpha},  \minzero}^{-1} e_{\minone\alpha}({ X};{\alpha}_{\minzero}) \right\}
    \end{bmatrix} 
    \end{pmatrix} 
    \\  \addlinespace[1.5mm] &=\E\begin{pmatrix}
        &  \begin{bmatrix}
    \dfrac{(-1)^{1-Z}}{\tilde e_{\minzero}(X,Z)}  \left\{Y- Q_{\minone}^\dr({ X},Z; {\beta})-{\mathcal{D}}_{\minone}^\dr( {\beta}) \mathcal{H}_{{\alpha} {\alpha},  \minzero}^{-1} e_{\minone\alpha}({ X};{\alpha}_{\minzero}) \right\} \\\addlinespace[1mm] +\{Q_{\minone}^\dr({ X},1; {\beta})-Q_{\minone}^\dr({ X},0; {\beta})\}
    \end{bmatrix}  \\\addlinespace[1mm]&\times \begin{bmatrix}
        \dfrac{Z-e_{\minzero}({ X}) }{e_{\minzero}({ X})}  \left\{ Q^\dr_{{\minone\beta}}({ X},1; {{\beta}})+{\mathcal{D}}_{\minone\beta}^\dr( {{\beta}}) {\mathcal{H}}_{{\alpha} {\alpha},  \minzero}^{-1} e_{\minone\alpha}({ X};{\alpha}_{\minzero}) \right\}\\\addlinespace[1mm]-\dfrac{ e_{\minzero}({ X})-Z }{ 1-e_{\minzero}({ X})}   \left\{ Q^\dr_{{\minone\beta}}({ X},0; {{\beta}})+{\mathcal{D}}_{\minone\beta}^\dr( {{\beta}}) {\mathcal{H}}_{{\alpha} {\alpha},  \minzero}^{-1} e_{\minone\alpha}({ X};{\alpha}_{\minzero}) \right\}
    \end{bmatrix} 
    \end{pmatrix} \\& =\E\begin{pmatrix}  \dfrac{(-1)^{1-Z}}{\tilde e_{\minzero}(X,Z)}  \left\{Y- Q_{\minone}^\dr({ X},Z; {\beta})-{\mathcal{D}}_{\minone}^\dr( {\beta}) \mathcal{H}_{{\alpha} {\alpha},  \minzero}^{-1} e_{\minone\alpha}({ X};{\alpha}_{\minzero}) \right\}
  \\\addlinespace[1mm] \times \begin{bmatrix}
        \dfrac{Z-e_{\minzero}({ X}) }{e_{\minzero}({ X})}  \left\{ Q^\dr_{{\minone\beta}}({ X},1; {{\beta}})+{\mathcal{D}}_{\minone\beta}^\dr( {{\beta}}) {\mathcal{H}}_{{\alpha} {\alpha},  \minzero}^{-1} e_{\minone\alpha}({ X};{\alpha}_{\minzero}) \right\}\\\addlinespace[1mm]-\dfrac{ e_{\minzero}({ X})-Z }{ 1-e_{\minzero}({ X})}   \left\{ Q^\dr_{{\minone\beta}}({ X},0; {{\beta}})+{\mathcal{D}}_{\minone\beta}^\dr( {{\beta}}) {\mathcal{H}}_{{\alpha} {\alpha},  \minzero}^{-1} e_{\minone\alpha}({ X};{\alpha}_{\minzero}) \right\}
    \end{bmatrix} 
    \end{pmatrix}  \\& =    0,
\end{align*}
where the final equality holds due to the fact $ \E \left\{Z-e_\minzero(X)\mid X\right\}=0$.
 \end{proof}
 \section{The properties of improved DR estimator when using the estimated propensity scores}
 \label{sec:dr-ipw-eps} 
 \begin{lemma}
    \label{lem:dr-epsest}We assume that \eqref{eq: if-est-dr} has a unique solution. 
     Under Assumptions \ref{assump: positivity} and \ref{assump: negative-control},  
the proposed estimator  has the following properties: 
\begin{enumerate}[label=(\roman*)]
\item  When either the propensity score or the outcome   model is correctly specified, ${ {\tau} ^{\dr}_{\minone}(\hat\alpha,\hat{\beta} ^{\dre})}$ is a consistent estimator.
\item When the propensity score model is correctly specified, ${ {\tau} ^{\dr}_{\minone}(\hat\alpha,\hat{\beta} ^{\dre})}$ achieves the smallest asymptotic variance among all estimators of the form in \eqref{eq: dr-X}.
\item When both models are correctly specified, ${ {\tau} ^{\dr}_{\minone}(\hat\alpha,\hat{\beta} ^{\dre})}$ achieves the semiparametric efficiency bound.
\end{enumerate}

\end{lemma}   
\begin{proof}
    (i) When the model for propensity score is correctly specified, by observing that $\hat{{\alpha}} \stackrel{p}{\rightarrow} {\alpha}_0$ and $\hat{\mathcal{H}}_{{{\alpha}} {\alpha}}$, $\hat{{\mathcal{D}}}^\dr_{\minone}({{\beta}})$, $\hat{{\mathcal{D}}}^\dr_{\minone\beta}({{\beta}})$ converge to $\mathcal{H}_{{{\alpha}} {\alpha},  \minzero}, {\mathcal{D}}_{\minone}^\dr({{\beta}}), {{\mathcal{D}}}^\dr_{\minone\beta}({{\beta}})$, respectively, it is straightforward to show that the left-hand side of the estimating equation converges in probability to the left-hand side of \eqref{eq: if-est-dr}. Thus, $\hat{\beta}^\dre\stackrel{p}{\rightarrow}  {\beta}^\dre$.  On the other hand, when the outcome model is correct but the propensity score is not, $\hat{{\alpha}} \stackrel{p}{\rightarrow} {\alpha}^*$, where $\alpha^\ast$ is   the probability limit of $\hat{\alpha}$ and $e_\minone(X,\alpha^\ast)\neq e_{\minzero}(X)$. The left-hand side of  the  estimating equation   \eqref{eq: if-est-dr} 
 converges to
     \begin{align} 
     \label{eq:var-dr-dr}
    \E&   \begin{pmatrix}  \dfrac{(-1)^{1-Z}}{\tilde e_{\minone} (X,Z; \alpha^\ast)}  \left\{Y- Q_{\minone}^\dr({ X},Z; {\beta})- {\mathcal{D}}_{\minone*}^\dr(\beta){\mathcal{H}}_{{\alpha} {\alpha}, *}^{-1} e_{\minone\alpha}({ X};{ \alpha^\ast}) \right\}
  \\\addlinespace[1mm] \times \begin{bmatrix}
        \dfrac{Z- e_{\minone} ({ X};{ \alpha^\ast})}{ e_{\minone} ({ X};{ \alpha^\ast})}  \left\{ Q^\dr_{{\minone\beta}}({ X},1; {{\beta}})+{\mathcal{D}}_{\minone\beta*} ^\dr({{\beta}}){\mathcal{H}}_{{\alpha} {\alpha}, *}^{-1} e_{\minone\alpha}({ X};{ \alpha^\ast}) \right\}\\\addlinespace[1mm]-\dfrac{  e_{\minone} ({ X};{ \alpha^\ast})-Z }{ 1- e_{\minone} ({ X};{ \alpha^\ast})}   \left\{ Q^\dr_{{\minone\beta}}({ X},0; {{\beta}})+{\mathcal{D}}_{\minone\beta*} ^\dr({{\beta}}){\mathcal{H}}_{{\alpha} {\alpha}, *}^{-1} e_{\minone\alpha}({ X};{ \alpha^\ast}) \right\}
    \end{bmatrix} 
    \end{pmatrix} =0 ,
\end{align}
where ${\mathcal{H}}_{{\alpha} {\alpha}, *}=\E\left\{S_{{{\alpha}}}({ X}, Z; {{\alpha}}^*) S_{{{\alpha}}}^{\T}({ X}, Z; {{\alpha}}^*)\right\}, {\mathcal{D}}_{\minone*}^\dr({{\beta}})=-\E\left\{\partial  {\varphi}^\dr_{\minone}(Y, { X}, Z; {{\alpha}}^*, {{\beta}}) / \partial {{\alpha}}^{\T}\right\}$, and ${\mathcal{D}}_{\minone\beta*} ^\dr({{\beta}})=-\E\left\{\partial^2  {\varphi}^\dr_{\minone}(Y, { X}, Z; {{\alpha}}^*, {{\beta}}) / \partial {{\alpha}}^{\T} \partial {{\beta}}\right\}$. By noting that ${\mathcal{D}}_{\minone*}^\dr({{\beta}}_0)=0$, \eqref{eq:var-dr-dr} equals 0 when ${{\beta}}={{\beta}}_0$. Therefore, we have $\hat{\beta}^\dre\stackrel{p}{\rightarrow} {{\beta}}_0$  and $Q_{\minone}(X,Z;{\beta}^\dre)=Q_{\minzero}(X,Z)$,\begin{align*}  \begin{aligned}
  {\tau}&^{\dr}_\minone(\alpha^\ast_{\minzero},\hat{{\beta}}^\dre )\\& =  \mathbb{P}_n
 \begin{bmatrix}
   \dfrac{ZY-\{{Z-e_{\minone}({ X};\alpha^\ast)}\}Q_{{{\minone}}}^{\dr}({ X},1;\hat{{\beta}}^\dre)}{e_\minone( { X};{\alpha^\ast})}  - \dfrac{(1- Z)Y-\{{e_\minone( { X};{\alpha^\ast})}-Z\}Q_{{{\minone}}}^{\dr}({ X},0;\hat{{\beta}}^\dre)}{ 1 - e_\minone( { X};{\alpha^\ast})}  
 \end{bmatrix}  \\
& =\mathbb{P}_n\{Q_{\minzero}( { X},1 )-Q_{\minzero}( { X},0)\} +\mathbb{P}_n \begin{bmatrix}
   \dfrac{Z\{Y-Q_{\minzero}( { X},1) \}}{e_\minone( { X};{\alpha^\ast})}  -   \dfrac{(1-Z)\{Y-Q_{\minzero}( { X},0 )\}}{1-e_\minone( { X};{\alpha^\ast})} 
 \end{bmatrix} +o_p(1) \\&=\E\{Q_{\minzero}( { X},1 )-Q_{\minzero}( { X},0)\} +\E \begin{bmatrix}
   \dfrac{Z\{Y-Q_{\minzero}( { X},1) \}}{e_\minone( { X};{\alpha^\ast})}  -   \dfrac{(1-Z)\{Y-Q_{\minzero}( { X},0 )\}}{1-e_\minone( { X};{\alpha^\ast})} 
 \end{bmatrix} +o_p(1)\\&= \tau+o_p(1).
\end{aligned} \end{align*}   

(ii) Since $\hat{{\beta}}^{\dre} \stackrel{p}{\rightarrow} {\beta}^{\dre}$ when the propensity score is correct, by definition of ${\beta}^{\dre}$, it is trivial that $ {\tau}^{\dr}_\minone(\hat\alpha ,\hat{{\beta}}^\dre )$ achieves the smallest variance among class of estimators \eqref{eq: dr-X}.

(iii) 
When both models are correctly specified, by \eqref{eq:ipw-var-ps} in the main text, we have,
\begin{equation*} 
\begin{aligned} 
  \Sigma_{\minone}^{\dre} ({\beta}^{\dre})& =\Sigma_{\minone}^{\drt} ({\beta}^{\dre})  -{{\mathcal{D}}}_{\minone}^\dr( {{\beta}} ^{\dre}) {\mathcal{H}}_{{\alpha} {\alpha},  \minzero}^{-1} {{\mathcal{D}}}_{\minone}^\T( {{\beta}}^{\dre} ) \\& =\Sigma_{\minone}^{\drt} ({\beta}^{\dre}) \\& =\Sigma_{\minone}^{\drt} ({\beta}_\minzero) ,
\end{aligned}
\end{equation*}
which also achieves the semiparametric efficiency bound as \eqref{eq:seb}.
\end{proof}

\section{The proof of Theorem \ref{thm:eps}}
\label{sec:prof-thm3}
 We first rewrite Theorem \ref{thm:eps} in the main text as follows:
  \begin{theorem} 
    \label{thm:tps-dr}Let $$    \theta_{\ast}^{\dre}=\{({ {\alpha}}^{\ast} )^ \T, 
( { \psi}_1^\ast )^ \T, \ldots, ({ \psi}_q^\ast )^ \T, ({ \phi}_1 ^\ast )^ \T, \ldots, ({ \phi}_q^\ast )^ \T,( \xi^{\ast})^ \T, {\beta}^{\dre},\tau\}^\T, $$ where $ \theta_{\ast}^{\dre}$ is  the solution to the following set of equations:
   \begin{gather*} 
\E\left\{S_{\alpha} ({ X},Z ;{ {\alpha}}^{\ast})\right\}={0}\\
{({ \psi}_1^\ast, \ldots, { \psi}_q^\ast)-\E\left\{S_{\alpha}({ X}, Z;{ {\alpha}}^{\ast}) S_{\alpha}^{\T}({ X}, Z;{ {\alpha}}^{\ast})\right\}=0,} \\
{ ({ \phi}_1^\ast, \ldots, { \phi}_q^\ast)+\E\left\{\partial^2 {\varphi} _\minone(Y,{ X},  Z;{ {\alpha}}^{\ast}, {\beta} ^\dre) / \partial {\alpha}^{\T} \partial {\beta}\right\}=0,} \\ 
 \xi^\ast +\E\{\partial  {\varphi}_\minone^\dr(Y,{ X},  Z;{ {\alpha}}^{\ast}, {\beta}^{\dre}) / \partial {\alpha}\}={0},   \\\addlinespace[1mm]   
      \E  \begin{pmatrix}  \dfrac{(-1)^{1-Z}}{\tilde e_{\minone} (X,Z; {\alpha}^{\ast})}  \left\{Y- Q_{\minone}^\dr({ X},Z; {\beta}^\dre)-(\xi^\ast ) ^\T ({ \psi}_1^\ast, \ldots, {\psi}_q^\ast)^{-1}e_{\minone\alpha}({ X};{ {\alpha}^{\ast}}) \right\}
  \\\addlinespace[1mm] \times \begin{bmatrix}
        \dfrac{Z- e_{\minone} ({ X};{ {\alpha}^{\ast}})}{ e_{\minone} ({ X};{ {\alpha}^{\ast}})}  \left\{ Q^\dr_{{\minone\beta}}({ X},1; {{\beta}^\dre})+( { \phi}_1^{\dre}, \ldots, { \phi}_q^{\dre})({ \psi}_1^\ast, \ldots, {\psi}_q^\ast)^{-1} e_{\minone\alpha}({ X};{ {\alpha}^{\ast}}) \right\}\\\addlinespace[1mm]-\dfrac{  e_{\minone} ({ X};{ {\alpha}^{\ast}})-Z }{ 1- e_{\minone} ({ X};{ {\alpha}^{\ast}})}   \left\{ Q^\dr_{{\minone\beta}}({ X},0; {{\beta}^\dre})+( { \phi}_1^{\dre}, \ldots, { \phi}_q^{\dre})({ \psi}_1^\ast, \ldots, {\psi}_q^\ast)^{-1} e_{\minone\alpha}({ X};{ {\alpha}^{\ast}}) \right\}
    \end{bmatrix} 
    \end{pmatrix} =0 .
\end{gather*}
Given Assumptions \ref{assump: positivity} and \ref{assump: negative-control},  when either the propensity score or the outcome model is correctly specified, we have
    \begin{gather*}
        \sqrt{n}\{{ {\tau}^{\dr}_{\minone}(  {\alpha}^\ast ,{\hat{\beta}}^\dre  )}-\tau\} \stackrel{d}{\rightarrow} N(0,{{\Lambda}^\dre_\minone}).
\end{gather*}
   where ${{\Lambda}^\dre_\minone}={{\Lambda}^\dre}(\theta_{\ast}^{\dre})$,  
\begin{gather*}
  {B}_1^{\dre}({{\theta}})=\E\begin{pmatrix}
      \dfrac{\partial  {   m}_1^\dre}{\partial\alpha^\T}&  \dfrac{\partial  {   m}_1^\dre}{\partial{{\psi }}_1^\T}&\ldots&  \dfrac{\partial  {   m}_1^\dre}{\partial{{\psi }}_q^\T}&  \dfrac{\partial  {   m}_1^\dre}{\partial\phi_1^\T}&\ldots&  \dfrac{\partial  {   m}_1^\dre}{\partial\phi_q^\T}&  \dfrac{\partial  {   m}_1^\dre}{\partial{{\xi}}^\T}&  \dfrac{\partial  {   m}_1^\dre}{\partial{ \beta^\T} }\\     \addlinespace[1mm]
      \dfrac{\partial  {   m}_2^\dre}{\partial\alpha^\T}&  \dfrac{\partial  {   m}_2^\dre}{\partial{{\psi }}_1^\T}&\ldots&  \dfrac{\partial  {   m}_2^\dre}{\partial{{\psi }}_q^\T}&  \dfrac{\partial  {   m}_2^\dre}{\partial\phi_1^\T}&\ldots&  \dfrac{\partial  {   m}_2^\dre}{\partial\phi_q^\T}&  \dfrac{\partial  {   m}_2^\dre}{\partial{{\xi}}^\T}&  \dfrac{\partial  {   m}_2^\dre}{\partial{ \beta^\T} }\\     \addlinespace[1mm]
      \dfrac{\partial  {   m}_3^\dre}{\partial\alpha^\T}&  \dfrac{\partial  {   m}_3^\dre}{\partial{{\psi }}_1^\T}&\ldots&  \dfrac{\partial  {   m}_3^\dre}{\partial{{\psi }}_q^\T}&  \dfrac{\partial  {   m}_3^\dre}{\partial\phi_1^\T}&\ldots&  \dfrac{\partial  {   m}_3^\dre}{\partial\phi_q^\T}&  \dfrac{\partial  {   m}_3^\dre}{\partial{{\xi}}^\T}&  \dfrac{\partial  {   m}_3^\dre}{\partial{ \beta^\T} }
  \end{pmatrix}\\  {B}_2^{\dre}({{\theta}})=\E\begin{pmatrix}
      \dfrac{\partial  {   m}_4^\dre}{\partial\alpha^\T}&  \dfrac{\partial  {    m}_4^\dre}{\partial{{\psi }}_1^\T}&\ldots&  \dfrac{\partial  {   m}_4^\dre}{\partial{{\psi }}_q^\T}&  \dfrac{\partial  {   m}_4^\dre}{\partial\phi_1^\T}&\ldots&  \dfrac{\partial  {    m}_4^\dre}{\partial\phi_q^\T}&  \dfrac{\partial  {   m}_4^\dre}{\partial{{\xi}}^\T}&  \dfrac{\partial  {   m}_4^\dre}{\partial{ \beta^\T} }\\    
  \end{pmatrix}\\{D}_{11}^\dre({{\theta}})=\E\begin{Bmatrix}
       {   m}_1^\dre({    m}_1^\dre)^\T&    {   m}_1^\dre({    m}_2^\dre)^\T&    {   m}_1^\dre({    m}_3^\dre)^\T \\      {   m}_2^\dre({    m}_1^\dre)^\T&    {   m}_2^\dre({    m}_2^\dre)^\T&    {   m}_2^\dre({    m}_3^\dre)^\T \\      {   m}_3^\dre({    m}_1^\dre)^\T&    {   m}_3^\dre({    m}_2^\dre)^\T&    {   m}_3^\dre({    m}_3^\dre)^\T 
  \end{Bmatrix},~~{D}_{12}^\dre({{\theta}})=\E\begin{pmatrix}
       {   m}_1^\dre {m}_4^\dre  \\      {   m}_2^\dre{m}_4^\dre  \\      {   m}_3^\dre{m}_4^\dre 
  \end{pmatrix},\\{D}_{22}^\dre({{\theta}})  =\E ( {  m}_4^\dre)^2,\\
   {{\Lambda}^\dre}(\theta  )={B}_2^{\dre}({{\theta}})\{{B}_1^{\dre}({{\theta}})\}^{-1}  D_{11}^{\dre}({\theta})\{{B}_1^{\dre}({{\theta}})\}^{-\T}\{{B}_2^{\dre}({{\theta}})\}^\T-2{B}_2^{\dre}({{\theta}}) \{{B}_1^{\dre}({{\theta}})\}^{-\T}{D}_{12}^{\dre}({{\theta}}) +{D}_{22}^{\dre}({{\theta}}).
  \end{gather*}  
  Here we suppress the notations by writing  ${m}_1 ^\dre\equiv    {   m}_1^\dre({ O}_i;{\theta})$ and $ {m}_2 ^\dre\equiv    {   m}_2^\dre({ O}_i;{\theta})$, where  
\begin{align*}
 { m}^{\dre }({ O};{\theta} )&=\left[\{   { m}^{\dre }_1({ O};{\theta} ) \}^\T,
\{   { m}^{\dre }_2({ O};{\theta} ) \}^\T,
\{   { m}^{\dre }_3({ O};{\theta} ) \}^\T,
\{   m^{\dre }_4({ O};{\theta} ) \}^\T \right]^\T,\\
     { m}^{\dre }_1(O ;{\theta} )&=  S_{\alpha}({ X}, Z ; {\alpha}) ,\\
{ m}^{\dre }_2 ( O;{\theta})&=\begin{bmatrix}
     {\xi  }+\partial  {\varphi}^\dr_{\minone}(Y, { X}, Z;{\alpha}, {\beta}) / \partial {\alpha} \\
 \psi_1-\operatorname{col}_1\left\{S_{\alpha}({ X}, Z; {\alpha}) S^{\T}_{\alpha}({ X}, Z; {\alpha})\right\} \\
\vdots \\
 \psi_q-\operatorname{col}_q\left\{S_{\alpha}({ X}, Z; {\alpha}) S^{\T}_{\alpha}({ X}, Z; {\alpha})\right\} \\
 {\phi}_1^\dr +\mathrm{col}_1\left\{\partial^2  {\varphi}^\dr_{\minone}(Y, { X}, Z;{\alpha}, {\beta}) / \partial { \alpha}_1 \partial {\beta}\right\} \\
\vdots \\
 {\phi}_q^\dr+\mathrm{col}_q\left\{\partial^2  {\varphi}^\dr_{\minone}(Y, { X}, Z;{\alpha}, {\beta}) / \partial { \alpha}_q \partial {\beta}\right\}
\end{bmatrix},\\
{ m}^\dre_3( O;{\theta})&=        \begin{pmatrix}  \dfrac{(-1)^{1-Z}}{\tilde e_{\minone} (X,Z; {\alpha})}  \left\{Y- Q_{\minone}^\dr({ X},Z; {\beta}^\dre)- \xi   ^\T ({ \psi}_1, \ldots, {\psi}_q)^{-1}e_{\minone\alpha}({ X};{ {\alpha}}) \right\}
  \\\addlinespace[1mm] \times \begin{bmatrix}
        \dfrac{Z- e_{\minone} ({ X};{ {\alpha}})}{ e_{\minone} ({ X};{ {\alpha}})}  \left\{ Q^\dr_{{\minone\beta}}({ X},1; {{\beta}^\dre})+( { \phi}_1^{\dre}, \ldots, { \phi}_q^{\dre})({ \psi}_1, \ldots, {\psi}_q)^{-1} e_{\minone\alpha}({ X};{ {\alpha}}) \right\}\\\addlinespace[1mm]-\dfrac{  e_{\minone} ({ X};{ {\alpha}})-Z }{ 1- e_{\minone} ({ X};{ {\alpha}})}   \left\{ Q^\dr_{{\minone\beta}}({ X},0; {{\beta}^\dre})+( { \phi}_1^{\dre}, \ldots, { \phi}_q^{\dre})({ \psi}_1, \ldots, {\psi}_q)^{-1} e_{\minone\alpha}({ X};{ {\alpha}}) \right\}
    \end{bmatrix} 
    \end{pmatrix},\\
    m_4^\dre( O;{\theta})&=  \dfrac{ZY-\{{Z-e_{\minone}( {{ X}};{\alpha})}\}Q_{\minone}^\dr( { X},1;{\beta})}{e_{\minone}( {{ X}};{\alpha})}  - \dfrac{(1- Z)Y-\{{e_{\minone}( {{ X}};{\alpha})}-Z\}Q_{\minone}^\dr( { X},0;{\beta})}{ 1 - e_{\minone}( {{ X}};{\alpha})}    -\tau.
\end{align*}

\end{theorem} 
 We  next present the proof of  Theorem \ref{thm:eps}  as follows:
\begin{proof}[The proof of Theorem \ref{thm:eps}]
   When the propensity score is correctly specified, the unknown parameters are  $ \theta_{\ast}^{\dre}=\{({ {\alpha}}^{\ast} )^ \T, 
( { \psi}_1^\ast )^ \T, \ldots, ({ \psi}_q^\ast )^ \T, ({ \phi}_1 ^\ast )^ \T, \ldots, ({ \phi}_q^\ast )^ \T,( \xi^{\ast})^ \T, {\beta}^{\dre}\}^\T $. The estimating equation is
\begin{align*}
 { m}^{\dre }({ O};{\theta} )=\left[
\{   { m}^{\dre }_1({ O};{\theta} ) \}^\T,
\{   { m}^{\dre }_2({ O};{\theta} ) \}^\T,
\{   { m}^{\dre }_3({ O};{\theta} ) \}^\T,
\{   m^{\dre }_4({ O};{\theta} ) \}^\T 
\right]^\T,
\end{align*}By M-estimation theory, 
\begin{align*}
\sqrt{n}(\hat{{{\theta}}}^{\dre}-{{\theta}}_\ast^{\dre}) \stackrel{D}{\rightarrow} N(0,\left\{{B}^{\dre}({{\theta}}_0^{\dre})\right\}^{-1} {D}^{\dre}({{\theta}}_0^{\dre})\left[\left\{{B}^{\dre}({{\theta}}_0^{\dre})\right\}^{-1}\right]^\T),
\end{align*}
where
\begin{align*}
{B}^{\dre}({{\theta}})&=\begin{Bmatrix}
    {B}_1^{\dre}({{\theta}}) & 0_{s \times 1} \\
{B}_2^{\dre}({{\theta}}) & -1
\end{Bmatrix}, \quad {D}^{\dre}({{\theta}})=\begin{bmatrix}
{D}_{11}^{\dre}({{\theta}}) & {D}_{12}^{\dre}({{\theta}}) \\
\left\{{D}_{12}^{\dre}({{\theta}})\right\}^\T & {D}_{22}^{\dre}({{\theta}})
\end{bmatrix}
\end{align*} with all the quantities defined in the theorem. The rest is by the preliminary in Section \ref{ssec:Preliminary}.
\end{proof}
 \section{The proof of \eqref{eq:ipw-var-ps}}
 \label{sec: proof-of-eq10}
\begin{proof} 
According to \eqref{if:dr-est-3},  the influence function can be expressed as follows:
        \begin{align} 
        \label{if:dr-est-2}\tilde\varphi^\dr_{\minone}(Y, Z, { X};{\alpha}_{\minzero}, {{\beta}}^*)&=\varphi^\dr_{\minone}(Y, Z, { X};{\alpha}_{\minzero}, {{\beta}}^*)-{\mathcal{D}}_{\minone}^\dr( {{\beta}}^*) \mathcal{H}_{{\alpha} {\alpha},  \minzero}^{-1} S_{\alpha}(Z, { X};{\alpha}_{\minzero}).
\end{align} 
When the propensity score is correctly specified, we have that 
    \begin{align*} 
    {\mathcal{D}}_{\minone}^\dr({\beta})= -
\E\left\{\partial \varphi^\dr_{\minone}(Y, Z, { X}; {\alpha}_{\minzero}, {\beta}) / \partial {\alpha}^\T\right\}  =\E\{\varphi^\dr_{\minone}(Y, Z, { X}; {\alpha}_{\minzero}, {\beta})S^{\T}_{\alpha}({ X},Z;{\alpha}_{\minzero})\}. 
\end{align*}   
The asymptotic variance can be calculated as follows:
\begin{align*}
   \Sigma^\dre_\minone (\beta) &= \mathrm{var}\{\tilde\varphi^\dr_{\minone}(Y, Z, { X};{\alpha}_{\minzero}, {{\beta}})\}\\&=\E\{\tilde\varphi^\dr_{\minone}(Y, Z, { X};{\alpha}_{\minzero}, {{\beta}})\}^2 \\&=\E\{ \varphi^\dr_{\minone}(Y, Z, { X};{\alpha}_{\minzero}, {{\beta}})\}^2-2  \E\big[\varphi^\dr _{\minone}(Y, Z, { X};{\alpha}_{\minzero}, {{\beta}})S_{\alpha} ^\T(Z, { X};{\alpha}_{\minzero}) \mathcal{H}_{{\alpha} {\alpha},  \minzero}^{-1}  {\mathcal{D}}_{\minone}^\T( {{\beta}})\big]\\&~~~~~~~~+\E\{{\mathcal{D}}_{\minone}^\dr( {{\beta}}) \mathcal{H}_{{\alpha} {\alpha},  \minzero}^{-1} S_{\alpha}(Z, { X};{\alpha}_{\minzero})\}^2 \\&=\Sigma^{\drt}_\minone ({\beta}) -{\mathcal{D}}_{\minone}^\dr( {{\beta}})\mathcal{H}_{{\alpha} {\alpha},  \minzero}^{-1}  {\mathcal{D}}_{\minone}^\T( {{\beta}}),
\end{align*}
where the third equality holds due to 
\begin{align*}
    \E\{{\mathcal{D}}_{\minone}^\dr&( {{\beta}}) \mathcal{H}_{{\alpha} {\alpha},  \minzero}^{-1} S_{\alpha}(Z, { X};{\alpha}_{\minzero})\}^2 \\&=    \E\left\{{\mathcal{D}}_{\minone}^\dr ( {{\beta}}) \mathcal{H}_{{\alpha} {\alpha},  \minzero}^{-1} S_{\alpha}(Z, { X};{\alpha}_{\minzero})S_{\alpha} ^\T(Z, { X};{\alpha}_{\minzero}) \mathcal{H}_{{\alpha} {\alpha},  \minzero}^{-1}  {\mathcal{D}}_{\minone}^\T ( {{\beta}})   \right\}\\&=    {\mathcal{D}}_{\minone}^\dr ( {{\beta}}) \mathcal{H}_{{\alpha} {\alpha},  \minzero}^{-1} \E\left\{S_{\alpha}(Z, { X};{\alpha}_{\minzero})S_{\alpha} ^\T(Z, { X};{\alpha}_{\minzero}) \right\}\mathcal{H}_{{\alpha} {\alpha},  \minzero}^{-1}  \{{\mathcal{D}}_{\minone}^\dr ( {{\beta}}) \} ^\T\\&={\mathcal{D}}_{\minone}^\dr({{\beta}}) \mathcal{H}_{{\alpha} {\alpha},  \minzero}^{-1}   {\mathcal{D}}_{\minone}^\T ( {{\beta}})   .\\
  \E\big\{\varphi^\dr_\minone&(Y, Z, { X};{\alpha}_{\minzero}, {{\beta}})S_{\alpha} ^\T(Z, { X};{\alpha}_{\minzero}) \mathcal{H}_{{\alpha} {\alpha},  \minzero}^{-1}  {\mathcal{D}}_{\minone}^\T( {{\beta}})\big\} \\&=   \E\big\{\varphi^\dr_\minone (Y, Z, { X};{\alpha}_{\minzero}, {{\beta}})S_{\alpha} ^\T(Z, { X};{\alpha}_{\minzero}) \big\} \mathcal{H}_{{\alpha} {\alpha},  \minzero}^{-1}  {\mathcal{D}}_{\minone}^\T( {{\beta}})\\&={\mathcal{D}}_{\minone}^\dr({{\beta}})\mathcal{H}_{{\alpha} {\alpha},  \minzero}^{-1}  {\mathcal{D}}_{\minone}^\T( {{\beta}}).
\end{align*}
\end{proof}
\section{The proof of Theorem  \ref{thm:com-dre}}
 \label{eq:proof-coro3}  
\subsection{Preliminaries}
\begin{fact}
  For any $\delta$, by direct calculations, we have that,
 \begin{gather*}
\begin{aligned}
        {\mathcal{D}}_{\mintwoone}^\dr( {\delta} )  &= -
\E\left\{\partial \varphi^\dr_{\mintwo} (Y, { X},{ W},Z ; {\alpha} _{\minzero}, \gamma_{\minzero}, {\delta} ) / \partial {\alpha}^\T\right\} \\& =\E\{\varphi^\dr_{\mintwo} (Y, { X},{ W},Z ; {\alpha} _{\minzero}, \gamma_{\minzero}, {\delta} )S_{\alpha}^\T({ X},{ W},Z;{\alpha} _{\minzero}, \gamma_{\minzero})\},
\end{aligned}\\
\begin{aligned}
        {\mathcal{D}}_{\mintwotwo}^\dr( {\delta} )  &= -
\E\left\{\partial\varphi^\dr_{\mintwo} (Y, { X},{ W},Z ; {\alpha} _{\minzero}, \gamma_{\minzero}, {\delta} ) / \partial {\alpha}^\T\right\}  \\&  =\E\{\varphi^\dr_{\mintwo} (Y, { X},{ W},Z ; {\alpha} _{\minzero}, \gamma_{\minzero}, {\delta} )S_{\gamma}^\T({ X},{ W},Z;{\alpha} _{\minzero}, \gamma_{\minzero})\}.
\end{aligned}
\end{gather*}
\end{fact}
   \begin{fact}
        When the propensity score is fully known, i.e., $({\alpha}^{\ast},{\gamma}^{\ast}) =({\alpha}_{\minzero},{\gamma}_{\minzero})$, but the outcome model $Q^\dr_{\mintwo}({ X},W, Z;{\delta})$ may or may not be,  given  $\hat{{\delta}}-{\delta}^\ast=O_p(n^{-1/2})$, 
  the influence function for ${ {\tau} ^{\dr}_{\minone}( \alpha_\minzero,{\gamma}_{\minzero},\hat{\delta} )}$ is $\varphi^\dr_{\minone}(Y, { X},W,Z;{ {\alpha}}_\minzero, \gamma_\minzero,{{\delta}}^*)$, where
  \begin{align*}    
\begin{aligned}
     \varphi^\dr _{\mintwo}(Y, Z,{ X},W; {\alpha} _{\minzero}, \gamma_{\minzero}, \delta )=  &\frac{ZY}{e_{\minzero}({ X})}-\frac{(1- Z)Y}{ 1-e_{\minzero}({ X})}-   \frac{Z-e_{\minzero}({ X})}{e_{\minzero}({ X})}   Q_{\minone}^\dr({ X},W,1; {{\delta}})
 \\&~~+\frac{e_{\minzero}({ X})-Z}{ 1-e_{\minzero}({ X})}   Q_{\minone}^\dr({ X},W,0; {{\delta}})-\tau . 
\end{aligned}
  \end{align*}
   \end{fact}

 \subsection{The proof of \eqref{est:var-comp2}}
 \label{eq:verf-var-comp}
 \begin{lemma}
 \label{lem:proof-var-comp}
 Let     \begin{gather*} 
\mathcal{H}_{\mintwo,0}=\begin{pmatrix}
    \mathcal{H}_{{\alpha}  {\alpha},  \minzero} & \mathcal{H}_{{\alpha} \gamma, \minzero} \\
\mathcal{H}_{\gamma {\alpha},  \minzero} & \mathcal{H}_{\gamma\gamma,0}
\end{pmatrix},  
\end{gather*} 
we assume the matrix $\mathcal{H}_{\mintwo,0}$ is invertible. 
Given Assumptions \ref{assump: positivity} and \ref{assump: negative-control},  when the propensity scores $e_\minone(X;\alpha)$ and $e_\mintwo(X,W;\alpha,\gamma)$  are correctly specified,   we have
$$     \Sigma_{\mintwo} ^{\dre}(  {\delta} )   =\Sigma_{\mintwo}^{\drt}  (   {\delta} ) -{\mathcal{D}}_{\mintwoone}^\dr( {\delta} )   \mathcal{H}_{{\alpha}  {\alpha},  \minzero}^{-1}  {\mathcal{D}}_{\mintwoone}^{\T}(  {\delta} ) -\mathcal{M}^\dr({\delta} ) ,$$
where $
\mathcal{M}^\dr (\delta )= \{ {\mathcal{D}}_{\mintwotwo}^\dr(\delta )-{\mathcal{D}}_{\mintwoone}^\dr(\delta )   \mathcal{H}_{{\alpha}  {\alpha}, \minzero}^{-1} \mathcal{H}_{{\alpha} \gamma, \minzero}\}\mathcal{F}^{-1}\{ {\mathcal{D}}_{\mintwotwo}^\dr(\delta )-{\mathcal{D}}_{\mintwoone}^\dr(\delta )   \mathcal{H}_{{\alpha}  {\alpha},  \minzero}^{-1} \mathcal{H}_{{\alpha} \gamma, \minzero}\}^\T . $\end{lemma}
\begin{proof}
According to \eqref{eq:ipw-var-ps} in the main text, the general form of  asymptotic variance $   \Sigma^{\dre}_\mintwo (\delta)$  can be expressed as follows:
\begin{align}
\label{eps-fps-expression}
       \Sigma^{\dre}_\mintwo (\delta)& =\Sigma^{\drt}  _\mintwo(\delta) -{\mathcal{D}}_{\mintwo}^\dr(\delta ) \mathcal{H}_{\mintwo,0}^{-1} {\mathcal{D}}_{\mintwo}^\T(\delta  ),
\end{align}
where   $  {\mathcal{D}}_{\mintwo}^\dr(\delta)  =\left\{{\mathcal{D}}_{\mintwoone}^\dr(\delta)  , {\mathcal{D}}_{\mintwotwo}^\dr(\delta)  \right\} $. Let $  \mathcal{F}  =  \mathcal{H}_{\gamma \gamma,0}-\mathcal{H}_{\gamma  {\alpha},  \minzero} \mathcal{H}_{{\alpha}  {\alpha},  \minzero}^{-1} \mathcal{H}_{{\alpha} \gamma, \minzero}   $ and $ \mathcal{F}_1  =\mathcal{H}_{{\alpha}  {\alpha},  \minzero}-\mathcal{H}_{{\alpha} \gamma, \minzero} \mathcal{H}_{\gamma \gamma,0}^{-1} \mathcal{H}_{\gamma  {\alpha},  \minzero}$. The explicit  expression for $\mathcal{H}_{\mintwo,0}^{-1}$ is given by:
    \begin{gather*}
\begin{aligned}
\mathcal{H}_{\mintwo,0}^{-1}& =\begin{pmatrix}
\mathcal{F}_1^{-1} & -\mathcal{H}_{{\alpha}{\alpha},  \minzero}^{-1} \mathcal{H}_{\gamma {\alpha},  \minzero} ^\T\mathcal{F}^\T \\
-\mathcal{F}\mathcal{H}_{\gamma {\alpha},  \minzero} \mathcal{H}_{{\alpha}{\alpha},  \minzero}^{-1}  &\mathcal{F}^{-1}
\end{pmatrix} .
\end{aligned}
\end{gather*} 
Moreover, through direct verification, we know that
    \begin{align*}      
        \mathcal{F}_1^{-1}&= (\mathcal{H}_{{\alpha}  {\alpha},  \minzero}-\mathcal{H}_{{\alpha} \gamma, \minzero} \mathcal{H}_{\gamma \gamma,0}^{-1} \mathcal{H}_{\gamma  {\alpha},  \minzero})^{-1} \\&=\mathcal{H}_{{\alpha}  {\alpha},  \minzero}^{-1}+\mathcal{H}_{{\alpha}  {\alpha},  \minzero}^{-1} \mathcal{H}_{{\alpha} \gamma, \minzero}(\mathcal{H}_{\gamma \gamma,0}-\mathcal{H}_{\gamma  {\alpha},  \minzero} \mathcal{H}_{{\alpha}  {\alpha},  \minzero}^{-1} \mathcal{H}_{{\alpha} \gamma, \minzero})^{-1} \mathcal{H}_{\gamma  {\alpha},  \minzero} \mathcal{H}_{{\alpha}  {\alpha},  \minzero}^{-1} \\&=\mathcal{H}_{{\alpha}  {\alpha},  \minzero}^{-1}+\mathcal{H}_{{\alpha}  {\alpha},  \minzero}^{-1} \mathcal{H}_{{\alpha} \gamma, \minzero}\mathcal{F}^{-1} \mathcal{H}_{\gamma  {\alpha},  \minzero} \mathcal{H}_{{\alpha}  {\alpha},  \minzero}^{-1} .
    \end{align*}  
We now calculate the asymptotic variance $   \Sigma_\mintwo^{\dre} (\delta)$ in \eqref{eps-fps-expression}, 
\begin{align*} 
   {\mathcal{D}}_{\mintwo}^\dr&(\delta ) \mathcal{H}_{\mintwo,0}^{-1} {\mathcal{D}}_{\mintwo}^\T(\delta  )\\
& =\left\{ {\mathcal{D}}_{\mintwoone}^\dr(\delta ),~ {\mathcal{D}}_{\mintwotwo}^\dr(\delta )\right\}\begin{bmatrix}
    \mathcal{F}_1^{-1}{\mathcal{D}}_{\mintwoone}^\T(\delta )-\mathcal{H}_{{\alpha}  {\alpha},  \minzero}^{-1} \mathcal{H}_{\gamma {\alpha},  \minzero}^{\T} \mathcal{F}^{-1} {\mathcal{D}}_{\mintwotwo}^\T(\delta ) \\
-\mathcal{F}^{-1} \mathcal{H}_{\gamma {\alpha},  \minzero} \mathcal{H}_{ {\alpha} {\alpha}}^{-1}{\mathcal{D}}_{\mintwoone}^\T(\delta )+\mathcal{F}^{-1} {\mathcal{D}}_{\mintwotwo}^\T(\delta )
\end{bmatrix} \\&= {\mathcal{D}}_{\mintwoone}^\dr(\delta )     \mathcal{F}_1^{-1}{\mathcal{D}}_{\mintwoone}^\T(\delta )-{\mathcal{D}}_{\mintwoone}^\dr(\delta )   \mathcal{H}_{{\alpha}  {\alpha},  \minzero}^{-1} \mathcal{H}_{\gamma {\alpha},  \minzero}^{\T} \mathcal{F}^{-1} {\mathcal{D}}_{\mintwotwo}^\T(\delta ) \\&~~~- {\mathcal{D}}_{\mintwotwo}^\dr(\delta )\mathcal{F}^{-1} \mathcal{H}_{\gamma {\alpha},  \minzero} \mathcal{H}_{ {\alpha} {\alpha}}^{-1}{\mathcal{D}}_{\mintwoone}^\T(\delta )+ {\mathcal{D}}_{\mintwotwo}^\dr(\delta )\mathcal{F}^{-1} {\mathcal{D}}_{\mintwotwo}^\T(\delta )\\&= {\mathcal{D}}_{\mintwoone}^\dr(\delta )   \mathcal{H}_{{\alpha}  {\alpha},  \minzero}^{-1} {\mathcal{D}}_{\mintwoone}^\T(\delta )+ {\mathcal{D}}_{\mintwoone}^\dr(\delta )   \mathcal{H}_{{\alpha}  {\alpha},  \minzero}^{-1} \mathcal{H}_{{\alpha} \gamma, \minzero}\mathcal{F}^{-1} \mathcal{H}_{\gamma  {\alpha},  \minzero} \mathcal{H}_{{\alpha}  {\alpha},  \minzero}^{-1}  {\mathcal{D}}_{\mintwoone}^\T(\delta )\\&~~~-{\mathcal{D}}_{\mintwoone}^\dr(\delta )   \mathcal{H}_{{\alpha}  {\alpha},  \minzero}^{-1} \mathcal{H}_{\gamma {\alpha},  \minzero}^{\T} \mathcal{F}^{-1} {\mathcal{D}}_{\mintwotwo}^\T(\delta ) - {\mathcal{D}}_{\mintwotwo}^\dr(\delta )\mathcal{F}^{-1} \mathcal{H}_{\gamma {\alpha},  \minzero} \mathcal{H}_{ {\alpha} {\alpha}}^{-1}{\mathcal{D}}_{\mintwoone}^\T(\delta )\\&~~~+ {\mathcal{D}}_{\mintwotwo}^\dr(\delta )\mathcal{F}^{-1} {\mathcal{D}}_{\mintwotwo}^\T(\delta )\\&= {\mathcal{D}}_{\mintwoone}^\dr(\delta )   \mathcal{H}_{{\alpha}  {\alpha},  \minzero}^{-1} {\mathcal{D}}_{\mintwoone}^\T(\delta )+ {\mathcal{D}}_{\mintwoone}^\dr(\delta )   \mathcal{H}_{{\alpha}  {\alpha},  \minzero}^{-1} \mathcal{H}_{{\alpha} \gamma, \minzero}\mathcal{F}^{-1} \mathcal{H}_{\gamma  {\alpha},  \minzero} \mathcal{H}_{{\alpha}  {\alpha},  \minzero}^{-1}  {\mathcal{D}}_{\mintwoone}^\T(\delta )\\&~~~- {\mathcal{D}}_{\mintwotwo}^\dr(\delta )\mathcal{F}^{-1} \mathcal{H}_{\gamma {\alpha},  \minzero} \mathcal{H}_{ {\alpha} {\alpha}}^{-1}{\mathcal{D}}_{\mintwoone}^\T(\delta ) + {\mathcal{D}}_{\mintwotwo}^\dr(\delta ) \mathcal{F}^{-1}  {\mathcal{D}}_{\mintwotwo}^\T(\delta )\\&~~~-{\mathcal{D}}_{\mintwoone}^\dr(\delta )   \mathcal{H}_{{\alpha}  {\alpha},  \minzero}^{-1} \mathcal{H}_{\gamma {\alpha},  \minzero}^{\T} \mathcal{F}^{-1} {\mathcal{D}}_{\mintwotwo}^\T(\delta )\\&= {\mathcal{D}}_{\mintwoone}^\dr(\delta )   \mathcal{H}_{{\alpha}  {\alpha},  \minzero}^{-1} {\mathcal{D}}_{\mintwoone}^\T(\delta ) -\{ {\mathcal{D}}_{\mintwotwo}^\dr(\delta )-{\mathcal{D}}_{\mintwoone}^\dr(\delta )   \mathcal{H}_{{\alpha}  {\alpha},  \minzero}^{-1} \mathcal{H}_{{\alpha} \gamma, \minzero}\}\mathcal{F}^{-1} \mathcal{H}_{\gamma {\alpha},  \minzero} \mathcal{H}_{ {\alpha} {\alpha}}^{-1}{\mathcal{D}}_{\mintwoone}^\T(\delta ) \\&~~~+\{ {\mathcal{D}}_{\mintwotwo}^\dr(\delta )-{\mathcal{D}}_{\mintwoone}^\dr(\delta )   \mathcal{H}_{{\alpha}  {\alpha},  \minzero}^{-1} \mathcal{H}_{\gamma {\alpha},  \minzero}^{\T} \}\mathcal{F}^{-1} {\mathcal{D}}_{\mintwotwo}^\T(\delta )\\&= {\mathcal{D}}_{\mintwoone}^\dr(\delta )   \mathcal{H}_{{\alpha}  {\alpha},  \minzero}^{-1} {\mathcal{D}}_{\mintwoone}^\T(\delta ) -\{ {\mathcal{D}}_{\mintwotwo}^\dr(\delta ) -{\mathcal{D}}_{\mintwoone}^\dr(\delta )   \mathcal{H}_{{\alpha}  {\alpha},  \minzero}^{-1} \mathcal{H}_{{\alpha} \gamma, \minzero}\}\mathcal{F}^{-1}\left\{ \mathcal{H}_{\gamma {\alpha},  \minzero} \mathcal{H}_{ {\alpha} {\alpha}}^{-1}{\mathcal{D}}_{\mintwoone}^\T(\delta )-{\mathcal{D}}_{\mintwotwo}^\T(\delta )\right\} .
\end{align*}
  We assume that the matrix $\mathcal{H}_{\mintwo,0}$ is positive definite. The inverse of $\mathcal{H}_{\mintwo,0}$ is also positive definite; therefore, the submatrix of $\mathcal{F}^{-1}$ is also a positive definite matrix.  We  have, $$\mathcal{M}(\delta)=\{ {\mathcal{D}}_{\mintwotwo}^\dr(\delta )-{\mathcal{D}}_{\mintwoone}^\dr(\delta )   \mathcal{H}_{{\alpha}  {\alpha}, \minzero}^{-1} \mathcal{H}_{{\alpha} \gamma, \minzero}\}\mathcal{F}^{-1}\{ {\mathcal{D}}_{\mintwotwo}^\dr(\delta )-{\mathcal{D}}_{\mintwoone}^\dr(\delta )   \mathcal{H}_{{\alpha}  {\alpha},  \minzero}^{-1} \mathcal{H}_{{\alpha} \gamma, \minzero}\}^\T \geq 0. $$ 
\end{proof}
\subsection{The formal proof of Theorem  \ref{thm:com-dre}}
\begin{proof} 
\label{proof:dr-var-comp-eps}
In this proof, we only  show that $  \Sigma_{\mintwo} ^{\dre}( {\delta} ^\dre)\leq   \Sigma_{\minone} ^{\dre}( {\beta}^\dre ) .$ 
The remaining inequalities can be derived by   Theorems \ref{thm:com-drt} and \eqref{eq:ipw-var-ps}.

(i) We first verify the inequality under Assumption \ref{cond:two-res-aipw}(i). According to Lemma \ref{lem:proof-var-comp},  the asymptotic variance based on the adjustment set $(X,W)$ can be expressed as follows: 
 \begin{align*}
       \Sigma^{\dre}_{\mintwo} (\delta)  &=\Sigma^{\drt}_{\mintwo}   (\delta) -{\mathcal{D}}_{\minone}^\dr(\delta ) \mathcal{H}_{\mintwo,0}^{-1} {\mathcal{D}}_{\minone}^\T(\delta  ) \\&=\Sigma^{\drt}_{\mintwo}   (\delta) -{\mathcal{D}}_{\mintwoone}^\dr(\delta )   \mathcal{H}_{{\alpha}  {\alpha},  \minzero}^{-1} {\mathcal{D}}_{\mintwoone}^\T(\delta ) -\mathcal{M}^\dr({\delta} ), 
\end{align*}
where $ \mathcal{M}^\dr (\delta ) \geq 0 $ for any $\delta$.   Let $\tilde\delta$   represent an extended vector as $\tilde{\delta} = \left\{({\beta} ^\dre)^\T, 0_{1\times m} \right\}^\T$ to accommodate the dimension of $\delta$, where $m$ is the dimension of the parameters involving predcitive covariates $W$ that need to be expanded.  By definition, we know two estimators $\tau_\minone(\alpha_\minzero,\beta^\dre)$  and $\tau_\mintwo(\alpha_\minzero,\gamma_\minzero,\tilde\delta)$  are   equivalent,  indicating that they should have the same asymptotic variance, i.e., $\Sigma_{\mintwo}^{\drt} (\tilde{\delta}) =\Sigma_{\minone}^{\drt} ( {\beta} ^\dre)$, where $\Sigma_{\minone}^{\drt} (\beta^\dre)$ is the asymptotic variance of $\tau_\minone(\alpha_\minzero,\beta^\dre)$, and $\Sigma_{\mintwo}^{\drt} (\tilde{\delta}) $ is the asymptotic variance of $\tau_\mintwo(\alpha_\minzero,\gamma_\minzero,\tilde{\delta})$.
Additionally, by direct verification of the definitions, we also have ${\mathcal{D}}_{ \minone}^\dr( {\beta} ^\dre)  ={\mathcal{D}}_{\mintwoone}^\dr(\tilde{\delta}) $. 
We then have, 
    \begin{align*}
     \Sigma^{\dre}_{\mintwo} (\tilde\delta )  &=\Sigma^{\drt} _{\mintwo}  (\tilde\delta) -{\mathcal{D}}_{\mintwoone}^\dr(\tilde\delta)   \mathcal{H}_{{\alpha}  {\alpha},  \minzero}^{-1}  {\mathcal{D}}_{\mintwoone}^{\T}(\tilde\delta) -\mathcal{M}^\dr({\tilde\delta} )\\&\leq  \Sigma^{\drt} _{\mintwo}  (\tilde\delta) -{\mathcal{D}}_{\mintwoone}^\dr(\tilde\delta)   \mathcal{H}_{{\alpha}  {\alpha},  \minzero}^{-1}  {\mathcal{D}}_{\mintwoone}^{\T}(\tilde\delta) \\& =\Sigma^{\drt} _{\minone}  ( \beta^\dre) -{\mathcal{D}}_{\minone}^\dr(\beta^\dre)   \mathcal{H}_{{\alpha}  {\alpha},  \minzero}^{-1}  {\mathcal{D}}_{\minone}^{\T}(\beta^\dre) \\&=  \Sigma^{\dre} _{\minone}  (\beta^\dre) ,
\end{align*} where the  final equality holds due to \eqref{eq:ipw-var-ps}. 
Finally, by the definition of ${\delta} ^\dre$, we have  $$
   \Sigma^{\dre} _{\minone} ({\beta} ^\dre) \geq   \Sigma^{\dre} _{\mintwo} ( \tilde\delta )\geq   \Sigma^{\dre}  _{\mintwo}( \delta^\dre ). $$

(ii) Under Assumption \ref{cond:two-res-aipw}(ii), we have $\delta^\drt=\delta^\dre=\delta_0$, where $\delta_0$ is the true value of the outcome model $Q_{\mintwo} (X,W,Z;\delta)$. According  to \eqref{eq:ipw-var-ps}, the asymptotic variances satisfy $\Sigma^\drt_{\mintwo} (\delta^\drt)=\Sigma^\dre_{\mintwo} ( \delta^\dre)$. The remaining proof follows the claims presented in the second part of Section \ref{sec:var-tps}.
\end{proof}

 \end{document}